\documentclass[conference,compsoc]{IEEEtran}

\IEEEoverridecommandlockouts

\usepackage{cite}
\usepackage{amsmath,amssymb,amsfonts,amsthm}
\usepackage{float}
\usepackage{graphicx}
\usepackage{enumitem}
\usepackage{textcomp}
\usepackage{xcolor}
\usepackage{hyperref}

\definecolor{setupColor}{HTML}{FFB570}
\definecolor{registerColor}{HTML}{67AB9F}
\definecolor{distributeColor}{HTML}{7EA6E0}
\definecolor{auditColor}{HTML}{EA6B66}
\definecolor{icrcdomain}{HTML}{7F00FF}
\definecolor{outsideActColor}{HTML}{666666}
\definecolor{showupt}{HTML}{E31A1C}
\definecolor{verifyentds}{HTML}{33A02C}
\definecolor{stationtophone}{HTML}{1F78B4}
\definecolor{phonetostation}{HTML}{A6CEE3}

\usepackage{xspace}

\usepackage{pifont}
\usepackage{tikz}
\usepackage{booktabs}

\usepackage{footnote}
\makesavenoteenv{tabular}
\makesavenoteenv{table}

\usepackage{comment}
\excludecomment{fullversion}
\includecomment{conferenceversion}

\newcommand{\code}[1]{\textsf{#1}}
\newcommand{\parBF}[1]{\noindent\textbf{#1.}}

\newcommand{\stepsym}[1]{\circled{\scriptsize#1}}
\newcommand{\step}[1]{step~\stepsym{#1}}
\newcommand{\circled}[1]{\tikz[baseline=(char.base)]{\node[shape=circle,draw,inner sep=0.8pt] (char) {#1};}}
\newcommand{\game}[1]{\noindent\textbf{Game \ensuremath{G_{#1}}.}}

\newcommand{\icrcAbbr}{ICRC\xspace}
\newcommand{\icrcFull}{the International Committee of the Red Cross\xspace}

\def\addvalue#1#2{\expandafter\gdef\csname my@data@#1\endcsname{#2}}\def\usevalue#1{\csname my@data@#1\endcsname}

\newcommand{\requirement}[3]{\addvalue{#2}{#1}\hypertarget{#2}\noindent\textit{\usevalue{#2}{\tiny #2}: #3.}
}
\newcommand{\reqlink}[1]{\hyperlink{#1}{\usevalue{#1}\tiny{#1}}}

\newcommand{\param}{\ensuremath{\code{param}}\xspace}
\newcommand{\GlobalSetup}[1]{\ensuremath{\code{GlobalSetup}(#1)}\xspace}
\newcommand{\OnlyNameGlobalSetup}{\ensuremath{\code{GlobalSetup}}\xspace}
\newcommand{\SetupT}[1]{\ensuremath{\code{SetupT}(#1)}\xspace}
\newcommand{\OnlyNameSetupT}{\ensuremath{\code{SetupT}}\xspace}
\newcommand{\secl}{\ensuremath{1^{\ell}}\xspace}

\newcommand{\skrs}{\ensuremath{\code{sk}}\xspace}
\newcommand{\pkrs}{\ensuremath{\code{pk}}\xspace}
\newcommand{\SetupRS}[1]{\ensuremath{\code{SetupRS}(#1)}\xspace}

\newcommand{\pk}{\ensuremath{\code{pk}}\xspace}
\newcommand{\sk}{\ensuremath{\code{sk}}\xspace}
\newcommand{\Prepare}[1]{\ensuremath{\code{PrepareRegT}(#1)}\xspace}
\newcommand{\OnlyNamePrepare}{\ensuremath{\code{PrepareRegT}}\xspace}
\newcommand{\req}{\ensuremath{\code{request}}\xspace}
\newcommand{\res}{\ensuremath{\code{response}}\xspace}
\newcommand{\Process}[1]{\ensuremath{\code{ProcessRegRS}(#1)}\xspace}
\newcommand{\OnlyNameProcess}{\ensuremath{\code{ProcessRegRS}}\xspace}
\newcommand{\Finish}[1]{\ensuremath{\code{FinishRegT}(#1)}\xspace}
\newcommand{\OnlyNameFinish}{\ensuremath{\code{FinishRegT}}\xspace}

\newcommand{\stateT}{\ensuremath{\code{st}_{T}}\xspace}
\newcommand{\BlockList}{\ensuremath{\code{BL}}\xspace}
\newcommand{\entproof}{\ensuremath{\pi_{\text{ent}}}\xspace}
\newcommand{\VerifyDS}[1]{\ensuremath{\code{VerifyEntDS}(#1)}\xspace}
\newcommand{\OnlyNameVerifyDS}{\ensuremath{\code{VerifyEntDS}}\xspace}
\newcommand{\AuditProof}{\ensuremath{\pi_{\textrm{aud}}}\xspace}
\newcommand{\VerifyA}[1]{\ensuremath{\code{VerifyAudA}(#1)}\xspace}
\newcommand{\OnlyNameVerifyA}{\ensuremath{\code{VerifyAudA}}\xspace}
\newcommand{\ShowupT}[1]{\ensuremath{\code{ShowupT}(#1)}\xspace}
\newcommand{\OnlyNameShowupT}{\ensuremath{\code{ShowupT}}\xspace}

\newcommand{\stateIss}{\code{st}_{\code{iss}}}

\newcommand{\CardWhispering}[1]{\ensuremath{\code{CardWhispering}(#1)}\xspace}

\newcommand{\parampc}{\ensuremath{\code{param}_{PC}}\xspace}
\newcommand{\comcommit}{\code{Commit}}
\newcommand{\comgen}{\code{Com.Gen}}

\newcommand{\group}{\ensuremath{\mathbb{G}}}
\newcommand{\grouporder}{q}
\newcommand{\generator}{g}
\newcommand{\Zq}{\mathbb{Z}_{\grouporder}}
\newcommand{\Zp}{\mathbb{Z}_{\grouporder}}
\newcommand{\pedersengenh}{h}
\newcommand{\randin}{\in_R}
\newcommand{\bins}{\{0, 1\}}

\newcommand{\paramps}{\ensuremath{\code{param}_{\text{PS}}}\xspace}
\newcommand{\skps}{\ensuremath{\code{sk}_{\text{PS}}}\xspace}
\newcommand{\pkps}{\ensuremath{\code{pk}_{\text{PS}}}\xspace}
\newcommand{\nizk}{\textsf{NIZK}}

\newcommand{\dssign}{\code{Sign}}
\newcommand{\dsgen}{\code{Gen}}
\newcommand{\dsverify}{\code{Verify}}

\newcommand{\OraclePrepareReg}{\ensuremath{\mathcal{O}_{\code{PrepareReg}}}\xspace}
\newcommand{\OracleFinishReg}{\ensuremath{\mathcal{O}_{\code{FinishReg}}}\xspace}
\newcommand{\OracleShowUp}{\ensuremath{\mathcal{O}_{\code{Showup}}}\xspace}
\newcommand{\OracleShowUpTwo}{\ensuremath{\mathcal{O}_{\code{ShowupTwo}}}\xspace}
\newcommand{\OracleHonestReg}{\ensuremath{\mathcal{O}_{\code{HonestReg}}}\xspace}
\newcommand{\OracleMalUserReg}{\ensuremath{{\mathcal{O}_{\code{MalUserReg}}}}\xspace}
\newcommand{\OracleVerifyEnt}{\ensuremath{\mathcal{O}_{\code{VerifyEnt}}}\xspace}

\newcommand{\expbase}[1]{\ensuremath{\code{Exp}^{\text{#1}}_{\adv}}\xspace}
\newcommand{\expbaseb}[2]{\ensuremath{\code{Exp}^{\text{#1}}_{\adv, #2}}\xspace}
\newcommand{\expind}[1][b]{\expbaseb{IND}{#1}}
\newcommand{\expent}[1][b]{\expbaseb{ENT}{#1}}
\newcommand{\expaud}{\expbase{AUD}}
\newcommand{\expsec}{\expbase{SEC}}

\newcommand{\hBL}{\ensuremath{\code{h}_{\text{BL}}}\xspace}

\newcommand{\prf}[2]{\ensuremath{\code{PRF}_{#1}(#2)}\xspace}
\newcommand{\PRFtwo}{\code{PRF}}
\newcommand{\signatureAudit}{\ensuremath{\signature_{\textrm{aud}}}\xspace}

\newcommand{\Verify}[1]{\ensuremath{\code{Verify}(#1)}\xspace}

\newcommand{\stateTid}{\ensuremath{\code{st}_{T}^{(\id)}}\xspace}

\newcommand{\EpoDic}{\ensuremath{\epsilon^{\text{last}}}\xspace}

\newcommand{\Rev}{\code{Rev}\xspace}
\newcommand{\Ent}{\code{Ent}\xspace}
\newcommand{\BlockLists}{\code{BlockLists}\xspace}

\newcommand{\cmark}{\ding{51}}
\newcommand{\xmark}{\ding{55}}

\newcommand{\adv}{\ensuremath{\mathcal{A}}\xspace}
\newcommand{\advb}{\ensuremath{\mathcal{B}}\xspace}

\newcommand{\entseen}{\ensuremath{\code{ent}_{\text{seen}}}\xspace}

\newcommand{\kH}{\ensuremath{\code{k}_H}\xspace}
\newcommand{\kHi}[1]{\code{k}_{H,#1}}
\newcommand{\RevokeId}{\ensuremath{\code{v}_{H}}\xspace}
\newcommand{\RevokeVal}{\ensuremath{r_{H}}\xspace}
\newcommand{\ent}{\ensuremath{\code{ent}_{H}}\xspace}

\newcommand{\com}{\ensuremath{\code{Com}_{\textrm{ent}}}\xspace}
\newcommand{\signature}{\ensuremath{\sigma}\xspace}

\newcommand{\tagH}{\ensuremath{\tau_H}\xspace}

\newcommand{\CreateAuditProofDS}[1]{\ensuremath{\code{GenAudPiDS}(#1)}\xspace}
\newcommand{\OnlyNameCreateAuditProofDS}{\code{GenAudPiDS}\xspace}

\newcommand{\entsum}{\ensuremath{\code{ent}_{\textrm{sum}}}\xspace}
\newcommand{\rsum}{\ensuremath{\code{r}_{\textrm{sum}}}\xspace}
\newcommand{\entmax}{\code{ent}_{\textrm{max}}}
\newcommand{\rreal}{r_{\textrm{real}}}

\newcommand{\epo}{\ensuremath{\epsilon}\xspace}
\newcommand{\epox}{\ensuremath{\epsilon^*}\xspace}

\newcommand{\id}{\ensuremath{\code{id}}\xspace}
\newcommand{\transcript}{\ensuremath{\code{log}}\xspace}

\newcommand{\idzero}{\ensuremath{\code{id}_{0}}\xspace}
\newcommand{\idone}{\ensuremath{\code{id}_{1}}\xspace}

\newcommand{\idb}{\ensuremath{\code{id}_{b}}\xspace}
\newcommand{\lastepoch}{\ensuremath{\epsilon_{\text{last}}}\xspace}

\newcommand{\prob}{\textsf{Pr}\!}
 
\newtheorem{theorem}{Theorem}
\theoremstyle{definition}
\newtheorem{definition}{Definition}

\usepackage[scaled=0.88]{helvet}

\usepackage[capitalise]{cleveref}

\usepackage{algorithm,algpseudocodex}

\begin{document}

\title{Not Yet Another Digital ID: Privacy-Preserving Humanitarian Aid Distribution\\[2mm]
\footnotesize{\normalfont\emph{This is the full version of the conference paper published at the IEEE Symposium on
  Security and Privacy 2023\cite{WangLSGT23}.\\ This version includes the
  unabridged proofs. Please cite the conference version.}}\vspace{-2mm}}

\author{
  \IEEEauthorblockN{
    Boya Wang\IEEEauthorrefmark{1},
    Wouter Lueks\IEEEauthorrefmark{2},
    Justinas Sukaitis\IEEEauthorrefmark{3},
    Vincent Graf Narbel\IEEEauthorrefmark{3},
    Carmela Troncoso\IEEEauthorrefmark{1}
  }
  \IEEEauthorblockA{
    \IEEEauthorrefmark{1}SPRING Lab, EPFL,
    Lausanne, Switzerland \\
    \{boya.wang,carmela.troncoso\}@epfl.ch
  }
  \IEEEauthorblockA{
    \IEEEauthorrefmark{2}CISPA Helmholtz Center for Information Security,
    Saarbr\"ucken, Germany \\
    lueks@cispa.de
  }
  \IEEEauthorblockA{
    \IEEEauthorrefmark{3}International Committee of the Red Cross,
    Geneva, Switzerland \\
    dpo@icrc.org
  }
}

\maketitle

\thispagestyle{plain}
\pagestyle{plain}

\begin{abstract}
Humanitarian aid-distribution programs help bring physical goods to people in need. Traditional paper-based solutions to support aid distribution do not scale to large populations and are hard to secure. Existing digital solutions solve these issues, at the cost of collecting large amount of personal information. This lack of privacy can endanger recipients' safety and harm their dignity. In collaboration with the International Committee of the Red Cross, we build a safe digital aid-distribution system. We first systematize the requirements such a system should satisfy. We then propose a decentralized solution based on the use of tokens that fulfills the needs of humanitarian organizations. It provides scalability and strong accountability, and, by design, guarantees the recipients' privacy. We provide two instantiations of our design, on a smart card and on a smartphone. We formally prove the security and privacy properties of these solutions, and empirically show that they can operate at scale.
\end{abstract}


\section{Introduction}

Humanitarian organizations, such as \icrcFull (\icrcAbbr)~\cite{IcrcMission09}, aim to protect and assist the victims of violence, famines, and disaster.
One of their main operations is the distribution of physical goods, such as food or blankets, in emergency scenarios~\cite{Icrcecosec20}.

Traditionally, humanitarian organizations use paper-based systems to support aid-distribution, e.g., a list with recipients' information and allocation of goods, or paper vouchers valid for particular aid items. 
These approaches, while practical, have important shortcomings: searching for information on a paper list does not scale beyond a few hundred recipients, vouchers are easily faked, etc.

To address these shortcomings, humanitarian organizations are looking into easy-to-scale digital solutions to support their aid-distribution programs.
They are also aware that digitalization should be handled with care, as it brings new risks to the vulnerable populations they serve~\cite{KaspersenICRC16}.

Building a digital aid-distribution system that preserves the safety, rights, and dignity of humanitarian aid recipients requires a deep understanding of the humanitarian context.
We partner with the \icrcAbbr to learn the requirements and constraints associated with distributing aid in emergencies.
Our interactions reveal the following challenges:

\begin{enumerate}[left=0pt]
\item \textit{Secure household-oriented aid.} Aid-distribution systems must permit aid allocation per household (i.e., a domestic unit of several members sharing meals and income), yet they must ensure that households can only request aid once per distribution round (e.g., per month).
\item \textit{Avoid reliance on powerful hardware and connectivity}. Most aid-distribution programs take place in crisis-affected settings where we cannot assume the existence of last-generation hardware or internet connectivity.
\item \textit{Auditability}. For accountability reasons, humanitarian organizations need to prove that aid is distributed in an honest manner, i.e., only to legitimate recipients.
\item \textit{Strong privacy.} Aid-distribution systems must avoid causing digital harm to the individuals~\cite{Noharm20}. The system must avoid generating databases with recipients' data and creating digital traces related to recipients' actions.
\end{enumerate}

Existing digital aid-distribution solutions can address the first three challenges.
Often, they achieve this by integrating an Identity Management System into their solution~\cite{Jointwfp15}.
This creates (central) databases with the personal data of recipients -- and even more sensitive information if the program requires strong authentication, such as the UN Refugee Agency's biometric identification system for refugees~\cite{Unhcr15}, or Pakistan's biometric-based Watan Card~\cite{PakistanWatanCard22}.

These solutions' reliance on data, however, conflicts with the strong need for privacy (challenge 4) and makes them ill-suited for the highly-sensitive humanitarian context~\cite{Aidingsur13}.
Not only can they jeopardize the safety of recipients~\cite{Afghanbiometrics22,Genocide01}, but they may also complicate the relationship of humanitarian organizations with local authorities, which shades the neutrality of the humanitarian actors~\cite{TAYLOR2015229}.
For example, in Yemen, the World Food Program clashed with Houthi authorities because of the disagreement over the usage and control of biometric data~\cite{YemenMaria21}.
Finally, from an ethical perspective, it is questionable whether gathering personal information of vulnerable people is acceptable given the risks that it entails for them~\cite{HayesICRC19,TAYLOR2015229}.

The reliance of data in existing systems is inherent to their approach to prevent distribution to illegitimate recipients and to ensure accountability, mainly based on centralizing the collection of logs. 
Privacy risks in such centralized solutions can only be avoided by using expensive, complex cryptography. 
To be able to address the four challenges simultaneously in a more efficient manner, we design an aid-distribution system that does not require centralizing data to achieve the desired properties.
Concretely, our contributions are the following: 
\begin{itemize}
\item We propose a token-based aid-distribution system that is secure, privacy-preserving, auditable, and that can operate with little to no connectivity. 
    \item We instantiate this system into two solutions that address the four challenges. 
    Both solutions (i) ensure that even if multiple tokens have been assigned to the same household, a household can receive aid only once per distribution round (challenge 1) 
    (ii) do not need to reveal any other information about households beyond their entitlement to request aid (challenge 4); 
    and (iii) produce privacy-preserving audit proofs to guarantee auditability (challenge 3).
    The first solution addresses challenge 2 by using smart cards as tokens, and the second solution reduces cost by using recipients' smartphones as tokens when aid is distributed in areas where such devices are available.
    \item For each solution, we prove that the protocols we propose fulfill the security and privacy requirements of the \icrcAbbr, and we empirically demonstrate that they can operate at scale.
    \item We discuss deployment considerations required to bring secure and privacy-preserving digital solutions to the humanitarian setting. 
\end{itemize}

\section{Aid Distribution in Humanitarian Context}

In this section, we describe how aid distribution works within the \icrcAbbr and elicit the requirements that a digitally supported aid-distribution system must fulfill.

\subsection{Requirements Gathering}
\label{sec:requirements}

To understand the needs of the \icrcAbbr, we worked closely with staff from the \icrcAbbr Data Protection Office.
We also organized two dedicated workshops with staff who are experienced in field operations, and held meetings with the staff in charge of organizing and coordinating aid-distribution programs.
We refined the requirements in weekly meetings held with the Data Protection Office for more than a year.

\vspace{1mm}\parBF{Functional requirements}
Humanitarian organizations require the following functionality for an aid-distribution system to be suitable in their context:

\requirement{F1}{household}{Distribution per household}
In some programs, humanitarian aid is allocated to individuals.
However, we find that often the entitlement to aid is decided based on the \emph{needs of households}~\cite{Icrcecosec20}.
To this end, at least one member of the household must register with the \icrcAbbr to collect aid (see requirement \hyperlink{robust}{D3\tiny{robust}} below).
In this paper, we focus on per-household distribution, as any household-oriented solution can be trivially adapted to per-individual entitlement by equating households with individuals.

\requirement{F2}{modify}{Entitlement modification after registration}
In general, the \icrcAbbr expects the allocation to a household to not change during the aid-distribution program.
Yet, aid recipients' needs may change through time, e.g., due to births, deaths, migration, or marriages.
The \icrcAbbr must be able to modify the allocation to a household at any moment.

\requirement{F3}{periodic}{Periodic distribution}
In many cases, one-time distribution is enough to respond to an emergency. 
However, some aid-distribution programs are long-term, and the distribution of aid is divided into several \emph{distribution periods}~\cite{IcrcAnnualAfrica21}.
Recipients receive aid once per period (e.g., getting five bags of rice per month). 
In this situation, it is desirable that a household, through their representative(s), can get aid periodically without having to register every time.

\vspace{1mm}\parBF{Deployment requirements} 
Humanitarian organizations operate in extreme settings. These conditions constrain the aid-distribution system design space, as well as the technologies that can be used.

\requirement{D1}{low}{Low-end hardware and sparse connectivity}
The \icrcAbbr may operate in areas where neither high-end hardware nor stable connectivity is available.
For example, recipients may not be able to afford high-end hardware, or the weather conditions may prevent the use of hardware that may be seen as a commodity in Western societies (e.g., in some locations the temperature or humidity can be too high for the latest models of smartphones). 
Therefore, aid-distribution systems should require only low-end hardware that is available in developing countries, be functional with little connectivity to the internet, and be reliable in austere environments.

\requirement{D2}{scale}{Efficiency at medium scale}
The \icrcAbbr distributes goods in unstable regions~\cite{Icrcecosec20}. 
Letting a lot of recipients wait for a long time in such areas may put at risk these recipients, as well as the \icrcAbbr staff and collaborators, e.g., creating a high-value target for terrorism.
In addition, queuing for basic human needs such as food may be questionable with respect to human dignity. Therefore, it is vital for the distribution system to be efficient even for numerous recipients. 
From our conversation with the \icrcAbbr, aid-distribution programs tend to involve thousands of recipients and, in some occasions, can involve more than 300,000 households per year. 

\requirement{D3}{robust}{Robust distribution}
Getting aid is critical for the people in need, so the system must guarantee distribution even when unexpected events occur. 
The \icrcAbbr staff reports two common situations in which robustness is needed. 
First, some recipients lose or accidentally damage the proof of registration and entitlement (e.g., a voucher or an aid-distribution card). 
Second, registered recipients may not be able to attend the distribution (e.g., in case of sickness or travel). 
The system should provide means for recipients to receive a new proof or registration, and it is desirable that more than one household member can collect the goods.

\requirement{D4}{usability}{Usability}
Aid distribution programs often take place in locations where there is a lack of digital literacy. 
This limits the type of solutions that humanitarian organizations can deploy, as not all technologies can be understood and used by the recipients.

We do not directly address this challenge in this work, but we choose smart cards and smartphone as platforms to develop our solution since both kind of devices have already been used by humanitarian organizations in the field.

\vspace{1mm}\parBF{Security requirements}
Humanitarian organizations are funded by voluntary contributions and there is no guarantee that such contributions will continue long-term~\cite{icrcfunding22}.
To ensure the continuity and the stability of fundraising, donors must trust the distribution system.
To promote trust, aid-distribution systems must include mechanisms to ensure that aid is only distributed to their legitimate recipients in the allocated quantity.

\requirement{S1}{limit}{Recipients should not be able to request more goods than they are entitled to} Legitimate recipients may try to get more aid than their household is entitled to (e.g., by requesting their entitlement more than once).
To ensure that the limited resources of the \icrcAbbr can be used to aid as many people as possible, the aid-distribution system should guarantee that recipients cannot request more goods than those established by the criteria.
We note that no technological solution can ensure that recipients do not request less goods (e.g., they may not attend the distribution).

\requirement{S2}{legitimate}{Only legitimate recipients should be able to obtain aid} Illegitimate recipients may try to get aid, e.g., by trying to impersonate a legitimate recipient.

\requirement{S3}{auditS}{Distribution must be auditable} To promote accountability of aid distribution programs, humanitarian organizations often keep track of their assistance activities by systematically collecting proof-of-delivery documents showing that goods are actually transferred from the organization to the recipient. Distribution tracking must enable auditors to find inconsistencies between the amounts distributed and the amounts delivered to eligible recipients in the field.

\vspace{1mm}\parBF{Privacy requirements} 
In the humanitarian context, recipients of aid rarely have any control or choice over which distribution system they use to get aid.
Any leakage of recipient's personal information, during operation of the program, or afterwards if the humanitarian organization is forced to leave the area without cleaning up the distribution system, can can put the recipients in danger. 
For example, systems built by Western actors in Afghanistan collected sensitive biometric data. 
In August 2021, when those actors had to leave the country, these sensitive data were left in the hands of the Taliban, raising concerns that these biometrics could be used to target political opponents~\cite{Afghanbiometrics22}. 
Hence, distribution systems in humanitarian settings must protect recipients' sensitive data to prevent any harm that could result from leaking these data.

\requirement{P1}{registration}{Privacy at registration}
Registration stations unavoidably learn some information about recipients to determine their eligibility to receive aid.
Such information could include place of residence, household demographics, economic situations as well as identities.
These necessary data cannot be hidden during registration. Yet, registration stations do not need to know transactional data of which recipients receive aid when and where.

\requirement{P2}{distribution}{Privacy at distribution}
Distribution stations must learn the eligibility and entitlement of a recipient in order to provide a recipient with the correct items.
Yet, because distribution stations are often run by third parties, the stations should not learn any other information about recipients beyond what is required to provide the assistance.
  
\requirement{P3}{auditP}{Privacy at auditing}
Auditors receive audit logs to verify the correct functioning of distribution stations. Such logs may contain sensitive information about recipients (e.g., when or how often a specific recipient received their goods).
To avoid endangering recipients, the system should ensure minimal disclosure of recipients' sensitive information to auditors. In most cases, it suffices to provide the total amount of distributed aid over a given time window and a proof that this aid was given to legitimate recipients.

\requirement{P4}{biometrics}{Privacy of biometrics}
Some aid-distribution systems use biometrics to ensure non-transferrability of aid. Biometric data are extremely privacy-sensitive. Because of privacy concerns, the \icrcAbbr's biometrics policy specifies that biometric templates must be stored on a device held by the data subject (i.e., aid recipient) whenever possible~\cite{Icrcbiopolicy19}. In particular, storing biometrics in central database held by the \icrcAbbr or service provider (e.g., registration or distribution stations) is considered as highly risky and must be avoided~\cite{Icrcbiopolicy19}. 

\subsection{Humanitarian Aid Distribution Workflow}

In this section, we provide a high-level description of a typical aid distribution workflow. We consider four actors whose interactions we illustrate in Fig.~\ref{fig:sys_view}:
\begin{itemize}
  \item \textbf{Recipients.} Recipients are those individuals that receive aid. We assume that recipients belong to one (and only one) household.
  We say an individual is a \textit{legitimate} recipient for a household if they are registered for getting aid of that household; and we say they are \textit{illegitimate} recipients if they are not registered in the system, or not associated to the household for which they receive aid.
  \item \textbf{Registration Station.} The registration station is in charge of enrolling legitimate recipients and determining their entitlement. Registration data are usually handled directly by the \icrcAbbr, or other trusted local parties. As such, the registration stations usually rely on the privileges and immunities afforded to the \icrcAbbr~\cite{BlondCTJFH18}.
  \item \textbf{Distribution Station.} Distribution stations are in charge of distributing aid to legitimate recipients. Distribution stations are often operated by third parties.
  \item \textbf{Auditors.} Auditors are potentially external parties that use audit records produced by distribution station to verify the honest behavior of these stations.
\end{itemize}

\begin{figure}[tbp]
  \centering
  \includegraphics[trim=3cm 5cm 3cm 5cm, width=0.8\columnwidth]{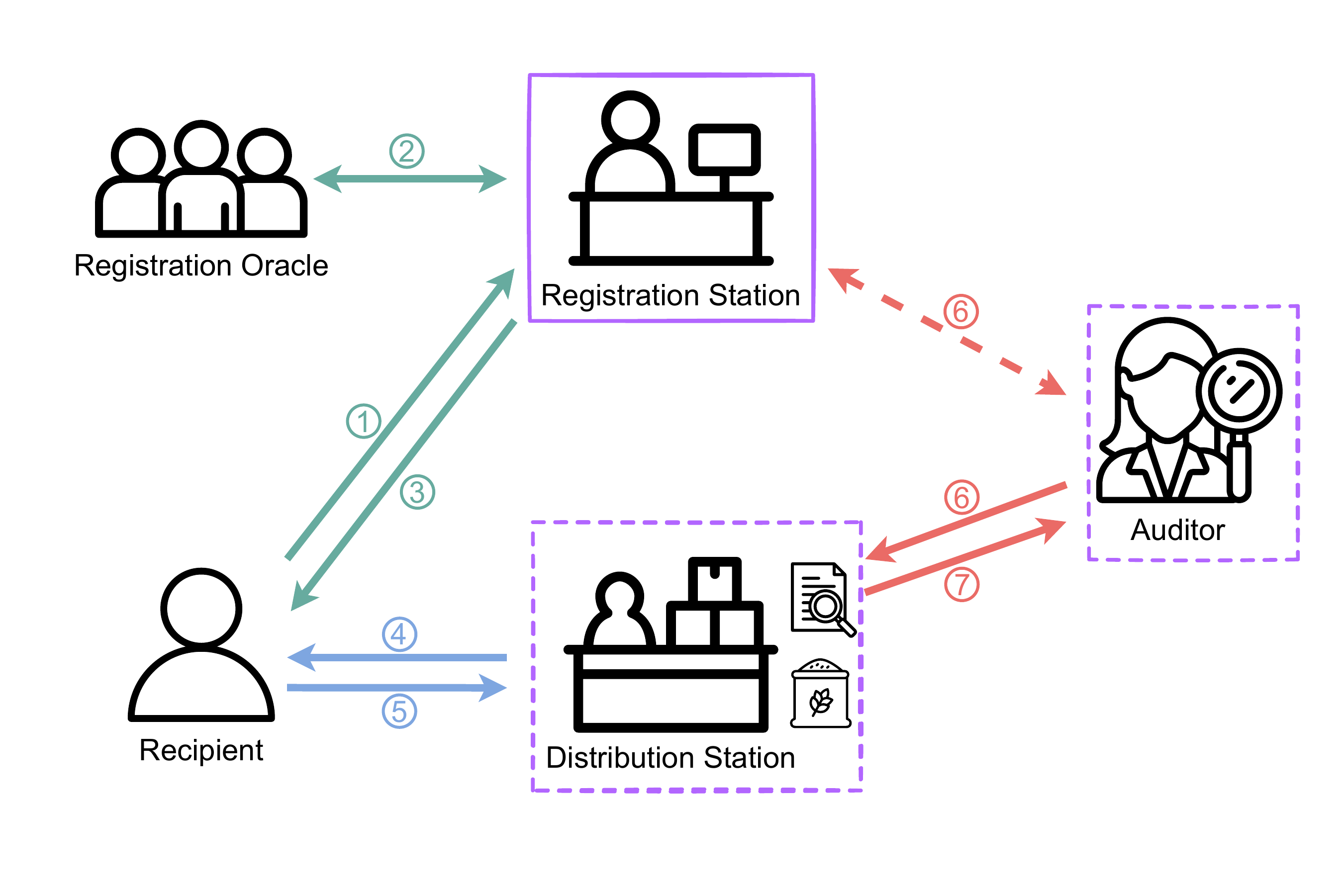}
\caption{Workflow of humanitarian aid-distribution. Recipients interact with the registration station. The registration station relies on a registration oracle to assess correctness of the information. Recipients then request aid from the distribution station. The auditor can ask the distribution station for transaction records. The humanitarian organization controls the registration station, while the distribution station and the auditor can be under the control of third parties. }
  \label{fig:sys_view}
\end{figure}

Given an emergency situation, the \icrcAbbr works with local communities to set the criteria for selecting households that are eligible for receiving aid.
Then, the organization registers the legitimate recipients for those households together with their \emph{entitlement} (e.g., five bags rice per month)~\cite{IcrcEcosecTarget} (Fig.~\ref{fig:sys_view}, \step{1}). 
Registration can happen, among others, by setting up a registration desk in the emergency area, or by having the \icrcAbbr staff visit households and register their entitlement. 
Regardless of where registration happens, we will refer to the process as happening at the registration station.
During registration, the registration station may need to validate the eligibility of households and their aid request, and that recipients are only registered for one household. This is typically done with the help of local communities (e.g., consulting with village elders)~\cite{Icrccashguide07, Icrcecosec20, Icrcevidence21, schwartz19, worldbankassess}. In this paper, we abstract the validation process as a registration oracle that can assign recipients to households and transfers this knowledge to the registration station (Fig.~\ref{fig:sys_view}, \step{2}).
After consulting the oracle, the registration station returns the result (e.g., a voucher, a card) to the recipient (Fig.~\ref{fig:sys_view}, \step{3}). 

Distribution of aid happens at \emph{distribution stations}. These may share the same location as the registration station.
In practice, there can be several registration and distribution stations.
For simplicity, in this paper we abstract them as a single station and discuss how to extend our solution to multiple sites in Sect.~\ref{sec:practical}.

Recipients request aid on behalf of the household at the distribution station (Fig.~\ref{fig:sys_view}, steps \stepsym{4} and \stepsym{5}).
The distribution station verifies that recipients are legitimate and hands over the goods according to the recipient's entitlement.

When necessary, auditors can use information from the registration and distribution stations to verify that the aid distribution was done honestly (Fig.~\ref{fig:sys_view}, steps \stepsym{6}  and \stepsym{7}).

\parBF{Threat model}
With respect to security of aid distribution (\reqlink{limit}, \reqlink{legitimate}) we assume that the registration and distribution stations are honest.
This assumption is essential.
A malicious registration station can simply register illegitimate recipients and change their entitlement, and a malicious distribution station could distribute however much aid to whomever it chooses.
We assume (potentially illegitimate) recipients are malicious and aim to violate these security properties.
For auditability (\reqlink{auditS}) we assume a malicious distribution station that tries to pass the audit check, but assume the registration station is honest.

With respect to privacy of registration data, we must trust the registration station (recall, these data are needed to verify eligibility).
All other parties -- distribution stations, auditors, and other recipients -- are considered malicious.
For privacy of transaction records and privacy of biometric data (\reqlink{biometrics}) we assume that all parties (except the recipient themselves) are malicious.

\subsection{Existing Systems and Limitations}

We review existing aid-distribution systems and evaluate to what extent they fulfill the above requirements.

\parBF{Paper-based} The \icrcAbbr reports that, currently, most aid-distribution systems are paper-based.
We identify two kinds of paper-based systems which differ in how the information used at distribution is stored.

In \emph{pen-and-list} solutions, the information about eligible households, their members, and their entitlement, is stored in one \emph{centralized} paper list during registration, and this list is passed to the distribution station.
When recipients show up at the distribution station to claim their goods, the staff checks their identity to find them on the list.
Recipients can sign their names as the evidence of distribution, producing a proof-of-delivery to be used for auditing purposes.

Pen-and-list solutions satisfy the functionality requirements (\reqlink{household}, \reqlink{modify}, \reqlink{periodic}).
However, the verification of recipients at the distribution station does not scale (\reqlink{scale}).
It can take long time to find a name among many others, especially when names are in languages for which the distribution staff are not native or when names are handwritten.
In addition, if the list is acquired by third parties at registration, or during or after the distribution, \emph{all} information about \emph{all} registered participants would be revealed, breaching the privacy requirements (\reqlink{distribution}, \reqlink{auditP}).

In \emph{pen-and-voucher} solutions, the information about eligibility, household membership, and entitlement is \emph{decentralized} in vouchers given to the recipients. 
The distribution station takes as input the information from recipients and, upon verification, it delivers the goods.
The staff at the distribution station can mark the voucher (e.g., by punching a hole) to show that the aid has been given in the corresponding month.

Compared to using the pen-and-list approach, which enables easy modification of entitlement, the decentralized nature of the pen-and-voucher solution means modification needs the assistance from recipients who may not be willing to cooperate (\reqlink{modify}).
Moreover, because information about distribution is only recorded on the voucher, auditing would require to physically gather vouchers that would act as proof-of-delivery (\reqlink{auditS}). 
Gathering voucher for auditing can be difficult, e.g., when using punch-cards that are carried by recipients.
Even when vouchers can be easily collected, there may be personal information written or printed on the voucher, e.g., to make sure only legitimate recipients can use them (\reqlink{legitimate}).
This information may breach privacy at auditing (\reqlink{auditP}).

\parBF{Digital solutions} Aiming to improve scalability and auditability of paper-based systems, humanitarian organizations are exploring digital solutions.
Digital solutions are easy to scale and can easily provide detailed logs for auditing. 
Depending on their implementation they can have different drawbacks.
Solutions that digitalize vouchers in a straightforward manner suffer from the same problem: anyone with a (digital) voucher can request the goods even though the voucher does not belong to this person (\reqlink{legitimate}). 
To solve this problem, the distribution station needs to authenticate the recipient. 
For this purpose, a lot of aid-distribution systems integrate Identity Management systems (IdM)~\cite{Aidingsur13} and strong authentication solutions such as biometrics. 
In this way, the distribution station can verify ownership before providing the goods. 
While this resolves the legitimacy requirement, IdM-based solutions do not fulfill the privacy requirements (\reqlink{registration}, \reqlink{distribution}, \reqlink{auditP}, \reqlink{biometrics}): recipients' behavior is linkable, the distribution station can identify recipients, auditing records can leak personal information, and the central biometric database becomes a single point of failure for privacy.

\section{A Token-based Aid-distribution System}\label{sec:sys}
Next, we outline our proposal for a digital humanitarian aid-distribution system that fulfills the requirements from the \icrcAbbr as shown in Sect.~\ref{sec:requirements}. 

\subsection{Design choices} 
\parBF{Decentralization of information into digital tokens}
To avoid requiring high levels of trust on a single entity that collects all recipients' information, we store recipients' eligibility, entitlement information, and potentially authentication-oriented data on digital tokens that stay with the recipient.
We trust tokens to keep this information private.

Using a token enables us to decentralize recipients' information.
This in turn enables the design of a system in which the distribution station does not obtain information that allows to re-identify, or track the movements of, recipients (\reqlink{distribution});
and in which we can create auditing information without the need to trust the distribution station (\reqlink{auditS}).
If the system requires the use of authentication information such as biometrics, these can be stored on the token, avoiding the creation of a centralized database (\reqlink{biometrics}).
We discuss the use of biometrics in Sect.~\ref{subsec:auth}.

The use of \emph{digital} tokens removes the need to manually check eligibility at distribution (e.g., finding names on a list or checking information on a paper voucher), increasing the efficiency of the system (\reqlink{scale}).

\parBF{Enabling offline use of tokens}
To fulfill \reqlink{low}, we designed our protocols to work offline. Tokens communicate locally with registration and distribution stations and do not require Internet access. Throughout this paper, we assume that the communication between tokens and stations is encrypted and authenticated. As this can be achieved using standard techniques on top of our protocols, e.g., during session establishment, we omit channel encryption and authentication from the description of our systems.

\parBF{Modifying information via revocation and re-issuance}
Modification of information in a decentralized system is cumbersome, as once information is stored on a recipient's token the registration station has no access to it.
Since modifications are rare (\reqlink{modify}), we solve this problem by revoking tokens and re-issuing new ones with up-to-date information.
We choose to implement blocklist-based revocation rather than allowlist-based revocation~\cite{CamenischL02} because an allowlist-based method requires the knowledge of all the valid credentials, which may not be feasible in the field. 
In our design, the token receives the blocklist from the distribution station.

\parBF{Preventing double dipping using household tags}
A household must not be able to obtain aid more than once per round (\reqlink{limit}). At first thought, the use of strong authentication at distribution to detect recipients requesting aid more than once may seem to address this requirement. However, this does not prevent different members of the household from requesting the household-allocated aid. To ensure that different members of the same household can not obtain aid more than once per round, tokens output a household- and round-specific pseudorandom tag which they provide to the distribution station upon aid collection. The distribution station stores these tags. If another household member attempts to collect aid they must reveal their household tag for the round. Thus, the distribution station can easily detect the double-dipping attempt. We note that this mechanism does not require identifying recipients.

\parBF{Non-forgeable auditing combined with crosschecking}
Recall that to detect whether there have been goods not actually transferred from the organization to the recipient, the auditor should be able to find inconsistencies between the audit records and the warehouse storage.
We model this by enabling the distribution station to prove the total amount of goods it was authorized to deliver to legitimate recipients. To prove this, we require tokens to output non-forgeable records at distribution that can be seen as proofs-of-delivery. These records can be aggregated, becoming a proof-of-delivery of the total amount of goods given out by the station. By comparing the aggregated total to the warehouse records, auditors can detect when the station distributes more goods than those that recipients asked for in the round (e.g., the station distributed also to illegitimate recipients or gave extra goods to legitimate recipients). 
When doing this comparison, auditors need to consider that numbers may not coincide for legitimate reasons (e.g., goods get broken before delivery, or some goods need to be distributed in an emergency with no time for proof collection). These cases cannot be addressed by technology and need to be resolved using analog means (e.g., with manual annotations) as in the current system.
We also note that our auditing solution cannot detect the transfer of goods between legitimate recipients.

\subsection{A Token-based Scheme}\label{subsec:token}

\begin{figure*}[tbp]
    \centering
    \includegraphics[trim=4cm 2.5cm 1.5cm 6cm, width=\textwidth]{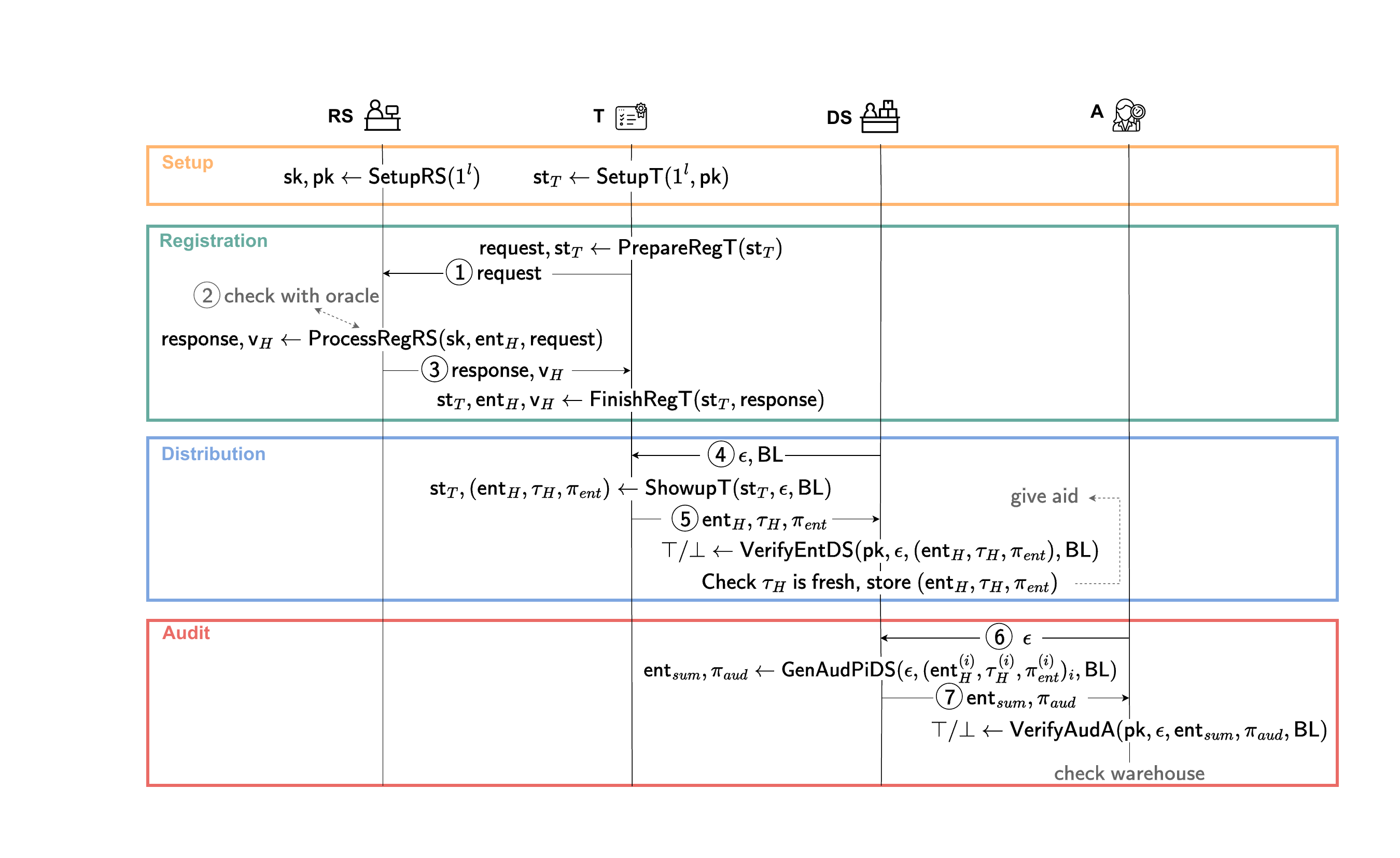}
    \caption{Overview of the token-based scheme omitting the global setup \OnlyNameGlobalSetup. From left to right, there are four parties: Registration Station (RS), Token (T), Distribution Station (DS), and Auditor (A). The boxes from top to bottom represent four phases: \textcolor{setupColor}{setup phase}, \textcolor{registerColor}{registration phase}, \textcolor{distributeColor}{distribution phase}, and \textcolor{auditColor}{auditing phase}. The circled numbers match steps from Fig.~\ref{fig:sys_view} to the scheme. \textcolor{outsideActColor}{Actions happening outside the scheme are marked in gray.}}
    \label{fig:scheme}
\end{figure*}

Our token-based scheme proceeds in four phases -- setup, registration, distribution, and auditing -- each consisting of several algorithms, whose syntax we describe below.
Fig.~\ref{fig:scheme} shows how parties use these algorithms in an aid-distribution system.
If a party aborts while running an algorithm, every field of the output is set to $\bot$.

\parBF{Setup Phase}
To set up the system, a trusted party generates cryptographic parameters that are shared by all parties.
The registration station generates a key pair $(\sk,\allowbreak \pk)$. Tokens are stateful, and token-specific algorithms update an explicit state $\stateT$ that is initialized during setup.
The algorithms satisfy the following syntax.
\begin{itemize}
    \item $\param \leftarrow \GlobalSetup{\secl}.$
    A trusted party takes as input the security parameter $\secl$, and returns the public parameters \param. The public parameters \param are implicit inputs to all other algorithms. 
In practice, these parameters would be obtained from well-known recommendations by experts.
    \item $\sk, \pk \leftarrow \SetupRS{\secl}.$
    The registration station takes as input the security parameter \secl, and outputs a private-public key pair $(\sk,\allowbreak \pk)$. The public key $\pk$ is known to all parties.
    \item $\stateT \leftarrow \SetupT{\secl, \pk}.$
    The token takes as input the security parameter \secl and the public key \pk. The token returns an initial state \stateT. 
\end{itemize}

\parBF{Registration Phase}
Recipients use a token to register in the system. This token can either be provided to them (e.g., a smart card) or be brought by the recipient (e.g., a smartphone).
To start registration, the token sends a registration request to the registration station together with a description of the household (Fig.~\ref{fig:scheme}, \step{1}).
If the registration oracle agrees that this household is eligible and the recipient is a member of the household, the oracle outputs an entitlement $\ent$ (Fig.~\ref{fig:scheme}, \step{2}). 
The registration station computes the response and sends it back to the token (Fig.~\ref{fig:scheme}, \step{3}).
The algorithms satisfy the following syntax.

\begin{itemize}
    \item $\req,\allowbreak \stateT \leftarrow \Prepare{\stateT}.$
    The token takes as input the internal state \stateT and outputs request information \req with a new state $\stateT$.
    \item $\res,\allowbreak \RevokeId \leftarrow \Process{\sk,\allowbreak \ent,\allowbreak \req}.$
    The registration station takes as input a private key \sk,
    the entitlement \ent of the household, and the request \req. It outputs a response \res, which includes information about entitlement, and a revocation value $\RevokeId$ for this household.
    \item $\stateT, \ent, \RevokeId \leftarrow \Finish{\stateT, \res}.$
    The token finishes registration by taking as input the token state \stateT and the registration response \res.
    It returns a new state $\stateT$, the entitlement \ent, and the revocation value $\RevokeId$. For simplicity, we assume that \ent is a scalar, but our scheme can trivially extend to vectors to support different types of goods.
\end{itemize}

The registration station stores the mapping between households and corresponding revocation values $\RevokeId$. To revoke a household's tokens, the registration station adds $\RevokeId$ to a public blocklist $\BlockList$.

\parBF{Distribution Phase} To receive aid, recipients go to the distribution station and start a request using their token. 
The station sends the current period $\epo$ and blocklist $\BlockList$ to the token (Fig.~\ref{fig:scheme}, \step{4}).
The token ensures the period $\epo$ is not in the past and verifies that its own revocation value is not in $\BlockList$. Then, the token sends the entitlement $\ent$, a period-specific tag $\tagH$, and a proof of entitlement $\entproof$ to the distribution station (Fig.~\ref{fig:scheme}, \step{5}).
The station checks the proof of entitlement to verify that $\ent$ is correct and that the token has not been blocked (\OnlyNameVerifyDS); and checks that the tag $\tagH$ has not been seen before.
If everything is correct, the station staff hands over the goods to the recipient, and the distribution station stores the tuple $(\ent, \tagH, \entproof)$ for auditing purposes.
The algorithms satisfy the following syntax. \looseness=-1

\begin{itemize}
    \item $\stateT,\allowbreak (\ent,\allowbreak \tagH,\allowbreak \entproof) \leftarrow \ShowupT{\stateT,\allowbreak \epo,\allowbreak \BlockList}.$
    The token takes as input its own state \stateT, the distribution period \epo, and a blocklist \BlockList.
    The algorithm outputs an updated state $\stateT$, the entitlement \ent, a tag \tagH that is unique for this household in this period, and a proof of the entitlement \entproof.
    \item $\top/\bot \leftarrow \VerifyDS{\pk, \epo, \allowbreak (\ent,\allowbreak \tagH,\allowbreak \entproof),\allowbreak \BlockList}.$
    The distribution station takes as input the public key $\pk$, the period $\epo$, the tuple $(\ent, \tagH, \entproof)$ from the token, and the blocklist \BlockList. The algorithm outputs $\top$ if \entproof verifies, $\bot$ otherwise.
\end{itemize}

\parBF{Audit Phase}
To audit transactions, the auditor requests an audit proof. The distribution station computes the transaction total $\entsum$ as well as a proof $\AuditProof$ and sends it to the auditor  (\OnlyNameCreateAuditProofDS and Fig.~\ref{fig:scheme}, \step{7}). The auditor uses the corresponding blocklist $\BlockList$ to verify the proof $\AuditProof$ (\OnlyNameVerifyA). The algorithms satisfy the following syntax.

\begin{itemize}
\item $\entsum,\allowbreak \AuditProof\leftarrow \CreateAuditProofDS{\epo,\allowbreak (\ent^{(i)},\allowbreak \tagH^{(i)},\allowbreak \entproof^{(i)})_i,\allowbreak \BlockList}.$
    The distribution station takes as input the period \epo, tuples $(\ent^{(i)}, \tagH^{(i)}, \entproof^{(i)})_{i}$, and a blocklist \BlockList.
    It outputs the total entitlement \entsum and an audit proof \AuditProof.
    \item $\top/\bot \leftarrow \VerifyA{\pk,\allowbreak \epo,\allowbreak \entsum,\allowbreak \AuditProof,\allowbreak \BlockList}.$
    The auditor takes as input the period \epo, the total entitlement \entsum, the audit proof \AuditProof, and the blocklist \BlockList that was used for these records.
    The algorithm outputs $\top$ if \AuditProof verifies, $\bot$ otherwise.
\end{itemize}

\section{Smart-card-based Aid Distribution}\label{sec:card}

We propose the first instantiation of the design in Sect.~\ref{sec:sys} which uses smart cards as tokens.

\parBF{Cheap, trustworthy tokens}
Smart cards are low-cost enough that humanitarian organizations can give one or multiple cards to a household (\reqlink{household}, \reqlink{low}, \reqlink{robust}).
Yet, smart cards have enough storage and computation capacity to perform the cryptographic operations we need. And they can be treated as trusted execution environments (TEEs), leading to a simple design. 
Specifically, we assume attackers cannot modify or observe computations inside the card. 

\parBF{Enabling privacy-friendly auditing}
Yet, the trusted nature of smart cards does not directly translate into privacy-friendly auditability (\reqlink{auditS}):
The distribution station must be able to convince the auditor that it only interacted with cards of legitimate receivers without violating recipient privacy.
The key idea in our design is that during distribution the smart card signs a \emph{homomorphic commitment to the entitlement}, and reveals its opening to the distribution station.
The distribution station can compute a commitment to the total entitlement obtained by homomorphic addition of all individual commitments.
The auditor can validate the signatures on the individual commitments, and check that the total entitlement matches the combined commitment, given an opening produced by the distribution station.

\parBF{Enabling multiple cards per household}
Preventing double dipping (\reqlink{limit}) with a \emph{single card} per household is easy -- the (trusted) card can remember which periods it successfully completed. However, for \reqlink{robust} our system requires multiple cards per household. 
To enable the detection of double dipping, cards output a per-household per-round tag $\tagH$ that they compute using a household secret $\kH$ and the current round $\epo$.
This household secret needs to be present in all the cards of the household.
We propose a \emph{card-whispering} protocol that enables ``cloning'' of cards.

\subsection{Preliminaries}

We use a digital signature scheme given by the algorithms $(\dsgen, \dssign,\dsverify)$. The key-generation algorithm $\dsgen(\secl)$ takes as input a security parameter $\secl$ and outputs a privacy-public key-pair $(\sk, \pk)$. The signing algorithm $\dssign(\sk, m)$ takes as input the private key $\sk$, a message $m$, and outputs a signature $\signature$. The verification algorithm $\dsverify(\pk, m, \signature)$ outputs $\top$ if $\signature$ is a valid signature on $m$ and $\bot$ otherwise.

We denote by $\PRFtwo_k: \{0, 1\}^{n} \to \{0, 1\}^{n}$ a pseudorandom function with key $k \in \{0, 1\}^n$. We assume the output of this function is indistinguishable from a random function.

We use the additively-homomorphic Pedersen commitment scheme~\cite{TorbenNoninteractive91}, defined by the algorithms $\comgen$ and $\comcommit$. The function $\comgen(\secl)$ takes as input the security parameter $\secl$ and outputs parameters $\parampc = (\group, \grouporder, \generator, \pedersengenh)$, where $\group$ is a cyclic group of prime order $\grouporder$ generated by $\generator$, and $\pedersengenh$ is another generator of $\group$ obtained by hashing a public string to a group element (this ensures that the discrete logarithm of $\pedersengenh$ with respect to $\generator$ is unknown to all parties). The function $\comcommit(m, r)$ takes as input a message $m \in \Zp$ and randomizer $r \in \Zp$ and outputs a commitment $c = \generator^m \pedersengenh^r$. This commitment scheme is additively homomorphic: $\comcommit(m_1, r_1) \cdot \comcommit(m_2, r_2) = \comcommit(m_1 + m_2, r_1 + r_2)$.

\subsection{Smartcard-based Aid Distribution Protocol} \label{sec:sc-based}
We define the smart card system by instantiating the scheme from Sect.~\ref{subsec:token}.

\parBF{Setup}
In the smart card system, the global setup (\OnlyNameGlobalSetup) outputs public parameters for the Pedersen commitment scheme and a hash function $H$.
The registration station generates a key pair $(\sk, \pk)$. The card initiates a last-seen-period counter necessary for \reqlink{distribution} (see distribution phase below).
The algorithms are implemented as follows:

\begin{itemize}
    \item $\param \leftarrow \GlobalSetup{\secl}.$
      The global setup function runs $\parampc \gets \comgen(\secl)$ to obtain Pedersen commitment parameters and picks a cryptographic hash function $H: \{0, 1\}^* \to \{0, 1\}^{\ell}$. It returns returns $\param = (\parampc, H)$.

    \item $\sk, \pk \leftarrow \SetupRS{\secl}.$ 
    The registration station takes as input \secl and runs $(\sk, \pk) \gets \dsgen(\secl)$ to obtain a signing key pair.

    \item $\stateT \leftarrow \SetupT{\secl,\allowbreak \pk}.$
    The token takes as input the security parameter $\secl$ and the public key \pk. Let the last-seen-period counter $\lastepoch = 0$. The function returns the state $\stateT = (\lastepoch,\allowbreak \pk)$.
\end{itemize}

\parBF{Registration}
To register the first recipient of a household, the registration station collaborates with the registration oracle to check eligibility and determine the household's entitlement. 
(To register additional members for a household see the card-whispering section below.)
The smart card does not send an initial $\req$ message, so we omit the definition of \OnlyNamePrepare.
\begin{itemize}
    \item $\res,\allowbreak \RevokeId \leftarrow \Process{\sk,\allowbreak \ent,\allowbreak \req}.$
    On input of the private key $\sk$ and the entitlement \ent (\req is null), the registration station creates a random revocation value $\RevokeId \randin \bins^{2\ell}$. It returns the response $\res = (\sk,\allowbreak \ent,\allowbreak \RevokeId)$. 

    \item $\stateT,\allowbreak \ent,\allowbreak \RevokeId \leftarrow \Finish{\stateT,\allowbreak \res}.$
On input of its state $\stateT = (\lastepoch, \pk)$ and the response \res from the registration station, the card parses the response as $\res = (\sk,\allowbreak \ent,\allowbreak \RevokeId)$ and checks that the private key \sk corresponds to the public key \pk (as set during \OnlyNameSetupT). The card generates a random household secret $\kH \randin \bins^{\ell}$ and returns the updated state $\stateT = (\lastepoch,\allowbreak \pk,\allowbreak \sk,\allowbreak \kH,\allowbreak \ent,\allowbreak \RevokeId)$.
\end{itemize}

\parBF{Card-whispering}\label{app:mul}
Recall that all cards assigned to the same household should be able to request aid (\reqlink{robust}).
To prevent double-dipping, all these cards must produce the same double-dipping tags (see distribution below).
Thus, when registering members of an already registered household we clone a previous registered household member's card.

To preserve recipient privacy (\reqlink{registration}, \reqlink{distribution}), the system must limit when, and by whom this protocol can be run.
Otherwise, an attacker with a cloned card can simply reproduce double-dipping tags $\tagH$ and thus track a household's activities.
We recommend that cloning can only happen \emph{after} the card to be cloned successfully authenticates its owner (e.g., using biometrics, see Sect.~\ref{subsec:auth}) and \emph{before} distribution starts. See \cref{app-cardwhisper} for the full protocol.

\parBF{Distribution}
We now describe how smart cards request aid, and how the distribution station verifies this request.
Smart cards use the $\lastepoch$ variable in their state to protect against malicious distribution stations by ensuring that they do not answer requests for past epochs (\reqlink{distribution}).

\begin{itemize}
    \item $\stateT,\allowbreak (\ent,\allowbreak \tagH,\allowbreak \entproof) \leftarrow \ShowupT{\stateT,\allowbreak \epo,\allowbreak \BlockList}.$
    The card takes as input its own state \stateT, the period \epo, and the latest blocklist \BlockList provided by the distribution station.
    It parses $\stateT$ as $(\lastepoch,\allowbreak \pk,\allowbreak \sk,\allowbreak \kH,\allowbreak \ent,\allowbreak \RevokeId)$.
    The card checks that it has not been blocked and that the period $\epo$ is not in the past, i.e., $\RevokeId \not\in \BlockList \land \lastepoch < \epo$.
    If either of the checks fails, the card aborts. Otherwise, the card computes a tag \tagH, a commitment \com, and a signature \signatureAudit,
    \begin{align*}
        \tagH &= \prf{\kH}{\epo} \\
        \com &= \code{Commit}(\ent, r), \quad r \randin \Zq \\
        \signatureAudit &= \code{Sign}(\sk, \tagH \parallel \epo \parallel \com \parallel \hBL),
    \end{align*}
    where $\hBL = H(\BlockList)$.
    Let $\entproof = (\signatureAudit,\allowbreak \com,\allowbreak r)$.
    It returns a new state \stateT by updating $\lastepoch$ to $\epo$,
    and $(\ent,\allowbreak \tagH,\allowbreak \entproof)$.

    \item $\top/\bot \leftarrow \VerifyDS{\pk, \epo, \allowbreak (\ent,\allowbreak \tagH,\allowbreak \entproof),\allowbreak \BlockList}.$
      The distribution station parses the proof $\entproof = (\signatureAudit,\allowbreak \com,\allowbreak r)$ and checks the signature $\signatureAudit$ by checking that $\Verify{\pk,\allowbreak \signatureAudit,\allowbreak \tagH\parallel \epo \parallel\com \parallel \hBL} = \top$ where $\hBL = H(\BlockList)$. It also verifies the commitment by checking that $\com = \comcommit(\ent, r)$. If both checks pass, it returns $\top$, and $\bot$ otherwise.
\end{itemize}
The station additionally checks that it has not seen $\tagH$ before handing out the goods (\reqlink{limit}).

\parBF{Auditing}
The distribution station can use the recorded transactions
$(\ent, \allowbreak \tagH, \allowbreak \entproof)$ to create audit proofs for the auditor.
We assume all records are valid for $\epo$ and $\BlockList$.
\begin{itemize}
    \item $\entsum,\allowbreak, \AuditProof\leftarrow \CreateAuditProofDS{\epo,\allowbreak (\ent^{(i)},\allowbreak \tagH^{(i)},\allowbreak \entproof^{(i)})_i, \BlockList}$
    The distribution station takes as input all the entries $(\ent^{(i)},\allowbreak \tagH^{(i)},\allowbreak \entproof^{(i)})$ in a distribution period \epo and lets
    \begin{equation*}
    \entproof^{(i)} = (\signatureAudit^{(i)}, \epo, \com^{(i)}, r^{(i)}).
    \end{equation*}
    The distribution station computes the sum of entitlement and the sum of random values
    \[
        \entsum =  \sum_{i=1}^{n_{\epo}}\ent^{(i)},\quad \rsum = \sum_{i=1}^{n_{\epo}}r^{(i)}.
    \]
    The station returns \entsum together with an audit proof
    \[
        \AuditProof = ( \rsum, (\signatureAudit^{(i)}, \tagH^{(i)}, \com^{(i)})_i, \hBL).
    \]
    where $\hBL = H(\BlockList)$.

    \item $\top/\bot \leftarrow \VerifyA{\pk,\allowbreak \epo,\allowbreak \entsum,\allowbreak \AuditProof,\allowbreak \BlockList}.$
    The auditor takes as input the public key \pk, the period \epo, the total entitlement \entsum, the auditing proof \AuditProof, and the blocklist \BlockList. Let $\AuditProof = (\rsum, (\signatureAudit^{(i)}, \tagH^{(i)}, \com^{(i)})_i, \hBL)$. The auditor checks:
    \begin{itemize}
        \item the validity of all signatures in the proof \AuditProof
        \[
        \top = \Verify{\pk, \signatureAudit^{(i)}, \tagH^{(i)}\parallel \epo \parallel\com^{(i)} \parallel \hBL},
        \]
        \item the uniqueness of all tags $\tagH^{(i)}$
        \item the sum of entitlement and the sum of random numbers match the commitment 
        \[ \prod_{i=1}^{n_{\epo}} \com^{(i)} = \comcommit(\entsum, \rsum) \] 
        \item that $\BlockList$ corresponds to $\hBL$, i.e., $\hBL = H(\BlockList)$.
    \end{itemize}
    If all four checks pass, it outputs $\top$ and $\bot$ otherwise.
\end{itemize}

\section{Smartphone Solution}\label{sec:phone}

In some areas where humanitarian organizations operate, smartphones are available (at least one per household).
To take advantage of this fact, we propose a second instantiation of the design in Sect.~\ref{sec:sys} using smartphones as tokens.

\parBF{Existing devices as tokens} Using smartphones as tokens reduces deployment costs associated with token distribution.
And smartphones have much more computation and storage capacity than smart cards. 
However, unlike a smart card, smartphones can run arbitrary software. Thus, their output cannot be trusted. 
We address this by using advanced cryptography in three ways.

\parBF{Proving eligibility}
We use attribute-based credentials (ABCs) to enable the phone to prove the recipient's eligibility and entitlement to the distribution station. 
The use of ABCs lets us satisfy~\reqlink{legitimate} while, at the same time, ensuring privacy~\reqlink{distribution}

We opt for the pairing-based ABC scheme by Pointcheval and Sanders (PS)~\cite{PointchevalS16}.
We considered the use of other ABC schemes, but found them to be less suited for our purposes.
Some schemes can only be shown once while maintaining unlinkability, and would thus require issuing one credential per distribution round~\cite{BaldimtsiL13,BrandsRethinking00}; others require the issuer and the verifier to share keys, which would harm auditability in our setting~\cite{ChaseMZ14}. 
There are alternatives offering the same functionality as PS credentials with no less complexity~\cite{AuSMC13,CamenischL02}.

Other technologies in previous research, e.g., FIDO~\cite{FIDO} and PrivacyPass~\cite{DavidsonGSTV18}, are not suitable in our use case. To be more specific, FIDO does not provide unlinkability towards the same relying party (the distribution station), i.e., whenever the recipient shows up twice at the same distribution station, the actions are linkable.
PrivacyPass could be used instead of ABCs by encoding entitlement into PrivacyPass tokens that are provided to the distribution station. However, it is hard to see how privacy pass tokens can support privacy-friendly auditability (which requires proving how many tokens were redeemed); or revocation (which requires privacy-friendly blocklisting).

\parBF{Double-dipping tags}
Instead of using a PRF to compute a double-dipping tag, we use a trick from direct anonymous attestation~\cite{BrickellCC04}: we compute $\tagH = H(\epo)^{\kH}$ where $\kH$ is the household secret and $H: \bins^{*} \to \group$ is a cryptographic hash function mapping strings to group elements. We extend the zero-knowledge proof of having a credential to prove the correct computation of $\tagH$.

\parBF{Proving non-revocation}
Finally, we rely on cryptography to let smartphones prove that their credential has not been revoked. 
Since we aim to provide privacy during distribution against malicious registration and distribution stations, not all revocation mechanisms are appropriate (see, e.g., Lapon et al.~\cite{LaponKDN11} for an overview). 
Because the revoking party, i.e., the registration station, is not trusted, phones should be able to detect whether they have been revoked. 
This rules out approaches such as verifier local revocation~\cite{BonehS04} as well as any approach based on verifiable encryption.

Instead, we use anonymous blocklisting~\cite{TsangAKS10,HenryG13} to let the phone prove that it is not currently revoked. 
These proofs are linear in the size of the blocklist. 
We considered dynamic accumulators~\cite{LiLX07,Nguyen05,DamgardT08a} -- whose non-revocation proofs are constant size -- but decided against them due to their increased cryptographic complexity and the need to handle witness updates in a privacy-friendly manner.

\subsection{Preliminaries: Attribute-based Credentials}
We use ABCs as a building block, and describe them only at a high level. During the issuance protocol, between a user and an issuer, the issuer provides a signed credential to the user. 
A credential contains one or more attributes, some of which can be hidden from the issuer during the signing process. 
We use $C(\kH, \ent, \RevokeVal)$ to denote a credential with the three attributes $\kH$, $\ent$ and $\RevokeVal$. 
Users can create a zero-knowledge proof of possessing a valid credential. For simplicity, we slightly abuse the notation and write:
\begin{equation*}
  \nizk\left\{ (C, \kH, \ent, \RevokeVal) : C(\kH, \ent, \RevokeVal) \right\}
\end{equation*}
to denote the non-interactive zero-knowledge proof of knowing a valid credential $C$ with attributes $\kH, \ent, \RevokeVal$. 
We refer to Appendix~\ref{app:ps-abc} for more details about ABCs.

\subsection{Smartphone-based Aid Distribution Protocol}
We define the smartphone-based system by instantiating the scheme from Sect.~\ref{subsec:token}.

\parBF{Setup}
The \OnlyNameGlobalSetup function determines global parameters, including those for the ABC scheme. 
The registration station acts as the issuer in the ABC scheme, and thus, creates a key pair for it.
\begin{itemize}
    \item $\param \leftarrow \GlobalSetup{\secl}.$
    The \OnlyNameGlobalSetup algorithm takes as input the security parameter \secl. It computes three types of parameters.
    First, it generates for the ABC scheme a type-III bilinear group pair given by $\paramps = (e(\cdot,\allowbreak \cdot),\allowbreak \group_1,\allowbreak \group_2,\allowbreak \group_T,\allowbreak \grouporder,\allowbreak g_1,\allowbreak g_2,\allowbreak g_T)$ where $g_1, g_2, g_T$ are generators of the groups $\group_1, \group_2, \group_T$ of order $\grouporder$ respectively.
    Second, it sets up the Pedersen commitment scheme in $\group_1$ with parameters $\parampc = (\group_1,\allowbreak \grouporder,\allowbreak g_1,\allowbreak h)$.
    Third, it picks a hash function $H: \{0, 1\}^{*} \rightarrow \mathbb{G}_{1}$ that maps strings onto the group $\mathbb{G}_{1}$.
    The application provider publishes all the global parameters $\param = (\paramps,\allowbreak \parampc,\allowbreak H)$.
  \item $\sk, \pk \leftarrow \SetupRS{\secl}.$
    The registration station generates a signing key pair $(\sk, \pk)$ for the ABC scheme.
  \item $\stateT \leftarrow \SetupT{\secl, \pk}.$
    The phone takes as input \secl and \pk from the registration station. 
    It initiates the last-seen-period counter \lastepoch and phone state $\stateT = (\lastepoch, \pk)$. 
\end{itemize}

It is essential for privacy that all phones receive the same public key \pk. In the following, we assume this is the case.

\parBF{Registration}
Recipients take their smartphones when registering at the registration station. The station associates legitimate recipients to the household with the help from the registration oracle. The station runs the ABC issuance protocol jointly with the phone to issue a credential.
\begin{itemize}
\item $\req,\allowbreak \stateT \leftarrow \Prepare{\stateT}.$
    The phone takes as input the initial state \stateT which it parses as $(\lastepoch, \pkrs)$.
    It generates a household secret $\kH \randin \Zq$. To enable revocation using anonymous blocklisting, the phone picks $\RevokeVal \randin \Zq$ and computes the revocation value $\RevokeId = (\generator_1, \generator_1^{\RevokeVal})$. The phone creates the first message $\req_{abc}$ in the ABC issuance protocol where $\kH$ and $\RevokeVal$ are phone-defined attributes. We assume that $\req_{\text{abc}}$ includes a proof that $\RevokeId$ is well formed (see~\cref{app:abc:issue}). The attribute $\kH$ and $\RevokeVal$ will be hidden from the registration station (i.e., the issuer). Let $\stateIss$ be the phone's private issuance state, $\stateT = (\lastepoch, \pkrs, \kH, \RevokeVal, \stateIss)$ the updated token state. The phone returns $\req = (\req_{\text{abc}}, \RevokeId)$ and $\stateT$.
    \item $\res,\allowbreak \RevokeId \leftarrow \Process{\sk,\allowbreak \ent,\allowbreak \req}.$
      On input of its private key \sk, the entitlement $\ent$ and the issue request $\req = (\req_{\text{abc}}, \RevokeId)$ from the phone, the registration station checks the validity of the request and that $\RevokeId$ has been correctly formed.
      Then it uses \sk and the phone's request $\req_{\text{abc}}$ to create a preliminary credential on the attributes $(\kH, \ent, \RevokeVal)$ (without learning either $\kH$ or $\RevokeVal$) which it packs into $\res$.
      It returns $\res$ and $\RevokeId$.
    \item $\stateT,\allowbreak \ent,\allowbreak \RevokeId \leftarrow \Finish{\stateT,\allowbreak \res}.$ On input of its internal state $\stateT = (\lastepoch, \pk, \kH, \RevokeVal, \stateIss)$ and the registration station's response $\res$, the phone completes the credential issuance protocol to obtain a credential $C$ on the attributes $(\kH, \ent, \RevokeVal)$ using its issuance state $\stateIss$ and $\res$. Then, it verifies that $C$ is a valid credential under $\pkrs$, and aborts otherwise. The phone returns the new state $\stateT = (\lastepoch, \pk, C, \kH, \ent, \RevokeVal)$, $\ent$, and $\RevokeId = (\generator_1, \generator_1^{\RevokeVal})$.
\end{itemize}

\parBF{Distribution} We next describe the algorithms used during aid collection.

\begin{itemize}
    \item $\stateT,\allowbreak (\ent,\allowbreak \tagH,\allowbreak \entproof) \leftarrow \ShowupT{\stateT,\allowbreak \epo,\allowbreak \BlockList}$.
    The phone takes as input \stateT, the period \epo, and the blocklist \BlockList.
    Let $\stateT =  (\lastepoch, \pk, C, \kH, \ent, \RevokeVal)$.
    The phone verifies that $\epo$ is not in the past, and that it has not made a distribution request in this epoch before, i.e., it checks that $\lastepoch < \epo$.
    Then, the phone checks that it is not blocked, i.e., that for all $(h, H) \in \BlockList$, $H \not= h^{\RevokeVal}$.
    If any check fails, the phone aborts. To continue, the phone computes a tag \tagH and a commitment \com:
    \begin{align*}
        \tagH &= H(\epo)^{\kH},  \\
        \com  &= \code{Commit}(\parampc, \ent, r), r \randin \Zq.
    \end{align*}
    Finally, the phone constructs the proof
    \begin{multline*}
      \pi_s = \nizk\big\{ (C, \kH, \ent, \RevokeVal, r) : \tagH = H(\epo)^{\kH} \land\\
       C(\kH, \ent, \RevokeVal) \quad\land\quad \com = \code{Commit}(\ent, r) \\
            \land \quad \forall (h,H) \in \BlockList: H \not= h^{\RevokeVal} \big\},
    \end{multline*}
    to show that it knows a valid credential $C$ on $\kH, \ent, \RevokeVal$, that $\tagH$ is correctly constructed, that $\com$ is correctly computed, and that it is not revoked (using the blocklisting protocol~\cite{HenryG13}). Let $\stateT = (\epo, \pkrs, C, \kH, \ent, \RevokeVal)$ be the new state and $\entproof = (\pi_s, \com, r)$.  It returns $\stateT$ and $(\ent, \tagH, \entproof)$.

    \item $\top/\bot \leftarrow \VerifyDS{\pk, \epo, (\ent,\allowbreak \tagH,\allowbreak \entproof),\allowbreak \BlockList}$. On input of the public key $\pkrs$, the period $\epo$, the tuple $(\ent, \tagH, \entproof)$, and the blocklist \BlockList, the distribution station proceeds as follows. First, it parses $\entproof$ as $(\pi_s, \com, r)$. Then, it verifies the proof $\pi_s$ using the public key $\pkrs$ and the provided blocklist \BlockList. If the proof verifies, it returns $\top$, and $\bot$ otherwise.
\end{itemize}

\parBF{Auditing}
Auditing proceeds similarly to the smart card system, except that the auditor checks zero-knowledge proofs. Again, we assume the records are valid for $\epo$ and \BlockList.
\begin{itemize}
  \item $\entsum,\allowbreak \AuditProof\leftarrow \CreateAuditProofDS{\epo,\allowbreak (\ent^{(i)},\allowbreak \tagH^{(i)},\allowbreak \entproof^{(i)})_{i}, \BlockList}.$
  The distribution station takes as input all the entries $(\ent^{(i)},\allowbreak \tagH^{(i)},\allowbreak \entproof^{(i)})$ in a distribution period \epo and lets
  \begin{equation*}
    \entproof^{(i)} = (\pi_s^{(i)}, \com^{(i)}, r^{(i)}).
  \end{equation*}
  The distribution station computes the sum of entitlement
  $\entsum =  \sum_{i=1}^{n_{\epo}}\ent^{(i)}$ and the randomizer $\rsum = \sum_{i=1}^{n_{\epo}}r^{(i)}$. It returns \entsum and an audit proof
  \begin{equation*}
  \AuditProof = ( \rsum, (\pi_s^{(i)}, \tagH^{(i)}, \com^{(i)})_i).
  \end{equation*}

  \item $\top/\bot \leftarrow \VerifyA{\pk,\allowbreak \epo,\allowbreak \entsum,\allowbreak \AuditProof,\allowbreak \BlockList}.$
  The auditor takes as input the period \epo, the total entitlement \entsum, the auditing proof \AuditProof, and the blocklist \BlockList that was used in \epo.
  Let $\AuditProof = ( \rsum, (\pi_s^{(i)}, \tagH^{(i)}, \com^{(i)})_i)$.
  The auditor checks the validity of all the proofs $\pi_s^{(i)}$ with respect to $\tagH^{(i)}$, $\com^{(i)}$ and the blocklist $\BlockList$; checks the uniqueness of all tags $\tagH^{(i)}$; and checks that the sum of entitlement and the sum of random numbers match the commitment: $\prod_{i} \com^{(i)} = \comcommit(\entsum, \rsum)$. If all checks pass, it outputs $\top$, and $\bot$ otherwise.
\end{itemize}
 
\begin{algorithm}[tbp]
    \caption{Registration Oracles}\label{alg:reg-oracles}
    \begin{algorithmic}[1]
      \small
      \Function{$\mathcal{O}_{\code{HonestReg}}$}{$\id, \RevokeId, \ent$}
      \State \textbf{if} $\id \in \mathcal{H}\cup\mathcal{M}$ \textbf{then} return $\bot$  
      \State $\stateTid \leftarrow \SetupT{\secl, \pk}$
        \State $\req, \stateTid \leftarrow \Prepare{\stateTid}$
        \State $\res,\allowbreak \RevokeId \leftarrow \Process{\sk,\allowbreak \ent,\allowbreak \req}$ 
        \State $\stateTid, \ent, \RevokeId \leftarrow \Finish{\stateTid, \res}$
        \State $\mathcal{H} \leftarrow \mathcal{H} \cup \{\id\}$, $\Ent[\id] \leftarrow \ent$, $\Rev[\id] \gets \RevokeId$
        \State return $\RevokeId$
      \EndFunction
      \vspace{1mm}

      \Function{$\mathcal{O}_{\code{MalUserReg}}$}{$\id, \ent, \req$}
      \State \textbf{if} $\id \in \mathcal{H}\cup\mathcal{M}$ \textbf{then} return $\bot$  
      \State $\res,\allowbreak \RevokeId \leftarrow \Process{\sk,\allowbreak \ent,\allowbreak \req}$ 
        \State $\mathcal{M} \leftarrow \mathcal{M} \cup \{\id\}$, $\Ent[\id] \leftarrow \ent$, $\Rev[\id] \gets \RevokeId$
        \State return $\res$
      \EndFunction
      \vspace{1mm}

      \Function{$\mathcal{O}_{\code{PrepareReg}}$}{$\id$}
      \State \textbf{if} $\id \in \mathcal{H}_{pre}\cup\mathcal{H}$ \textbf{then} return $\bot$  
      \State $\stateTid \leftarrow \SetupT{\secl, \pk}$
        \State $\req, \stateTid \leftarrow \Prepare{\stateTid}$
        \State $\mathcal{H}_{pre} \leftarrow \mathcal{H}_{pre} \cup \{\id\}$
        \State return $\req$
      \EndFunction
      \vspace{1mm}

      \Function{$\mathcal{O}_{\code{FinishReg}}$}{$\id, \res$}
        \State if $\id \notin \mathcal{H}_{pre} \vee \id \in \mathcal{H}$ return $\bot$
        \State $\stateTid, \ent, \RevokeId \leftarrow \Finish{\stateTid, \res}$
        \State $\mathcal{H} \leftarrow \mathcal{H} \cup \{\id\}$, $\Ent[\id] \leftarrow \ent$, $\Rev[\id] \gets \RevokeId$
      \EndFunction
      \vspace{1mm}
    \end{algorithmic}
\end{algorithm}

\section{Properties and Proofs}\label{sec:pro}
We formalize the security and privacy properties needed to fulfill the requirements described in Sect.~\ref{sec:requirements} using cryptographic games. In the games, the adversary interacts with users and parties in the system using oracles. See \cref{alg:reg-oracles,alg:dist-oracles} for an overview of the oracles we use.
Oracles track information, e.g., the entitlement and revocation values assigned to recipients (in $\Ent, \Rev$), the honest and malicious recipients (in $\mathcal{H}, \mathcal{M}$), as well as the last epoch in which an honest recipient sent information to a distribution station (using $\lastepoch$). 
We use these to define win conditions.

\begin{algorithm}[tbp]
  \caption{Distribution Oracles}\label{alg:dist-oracles}
  \begin{algorithmic}[1]
    \small
      \Function{$\mathcal{O}_{\code{Showup}}$}{$\id, \epo, \BlockList$}
        \State \textbf{if} $\id \notin \mathcal{H}$ \textbf{return} $\bot$ 
        \State \textbf{if} $\Rev[\id] \not\in \BlockList$ \textbf{then} $\entsum[\epo,\BlockList] \gets \entsum[\epo,\BlockList] + \Ent[\id]$
        \State $\EpoDic[\id] \leftarrow \max(\EpoDic[\id], \epo)$ \label{alg:doracle:track-epoch}
        \State $\stateTid,\allowbreak (\ent,\allowbreak \tagH,\allowbreak \entproof) \leftarrow \ShowupT{\stateTid,\allowbreak \epo,\allowbreak \BlockList}$
\State return \ent, \tagH, \entproof
      \EndFunction
      \vspace{1mm}

      \Function{$\mathcal{O}_{\code{ShowupTwo}}$}{\idzero, \idone, \epo, \BlockList}
        \State if $\idzero, \idone \notin \mathcal{H}$ return $\bot$
        \State $\BlockLists[\epo] \gets \BlockLists[\epo] \cup \{ \BlockList \}$
        \State $\stateT^{(\idzero)},\allowbreak (\ent^{(0)},\allowbreak \tagH^{(0)},\allowbreak \entproof^{(0)}) \leftarrow \ShowupT{\stateT^{(\idzero)},\allowbreak \epo,\allowbreak \BlockList}$ 
        \If {$\VerifyDS{\pk, \epo, (\ent^{(0)},\allowbreak \tagH^{(0)},\allowbreak \entproof^{(0)}),\allowbreak \BlockList} \land \tagH^{(0)} \not\in \transcript^0$}
        \State $\transcript^0[\epo] \leftarrow \transcript^0[\epo] \cup \{(\ent^{(0)}, \tagH^{(0)}, \entproof^{(0)})\}$
        \EndIf
        \State $\stateT^{(\idone)},\allowbreak (\ent^{(1)},\allowbreak \tagH^{(1)},\allowbreak \entproof^{(1)}) \leftarrow \ShowupT{\stateT^{(\idone)},\allowbreak \epo,\allowbreak \BlockList}$ 
        \If {$\VerifyDS{\pk, \epo, (\ent^{(1)},\allowbreak \tagH^{(1)},\allowbreak \entproof^{(1)}),\allowbreak \BlockList} \land \tagH^{(1)} \not\in \transcript^{1}$}
        \State $\transcript^1[\epo] \leftarrow \transcript^1[\epo] \cup \{(\ent^{(1)}, \tagH^{(1)}, \entproof^{(1)})\}$
        \EndIf
      \EndFunction
  \end{algorithmic}
\end{algorithm}

\subsection{Privacy of Distribution}
An aid-distribution system should preserve the privacy of recipients at distribution against (potentially malicious) registration and distribution stations (\reqlink{registration}, \reqlink{distribution}). We model this using the indistinguishability experiment in Algorithm~\ref{alg:ind}.
Adversary $\adv$ plays the role of a malicious registration station and produces the corresponding public key $\pkrs$ (line~\ref{alg:ind:output-pk}). The adversary can use the oracles $\OraclePrepareReg$ and $\OracleFinishReg$ to interact as a (malicious) registration station with honest users (line~\ref{alg:ind:query-phase}). It can also act as a (malicious) distribution station and use $\OracleShowUp$ to request that an honest user runs $\OnlyNameShowupT$. The oracle keeps track of the last epoch for which $\OracleShowUp$ has been called for each user (line~\ref{alg:doracle:track-epoch}, Algorithm~\ref{alg:dist-oracles}).
The modeling assumes that all tokens receive the same public key during \OnlyNameSetupT. In a real deployment, this assumption must be realized, e.g., by a trusted manufacturer (for cards) or application provider (for phones).

\begin{algorithm}[tbp]
\caption{Indistinguishability experiment}\label{alg:ind}
\begin{algorithmic}[1]
  \small
    \Function{$\expind$}{$\secl$}
      \State $\param \leftarrow \GlobalSetup{\secl}$
      \State $\pkrs \leftarrow \adv(\param)$ \label{alg:ind:output-pk}
      \State $\idzero, \idone, \epo^{*}, \BlockList^{*} \leftarrow \adv^{\mathcal{O}_{\code{PrepareReg}}(\cdot),\allowbreak \mathcal{O}_{\code{FinishReg}}(\cdot),\allowbreak \mathcal{O}_{\code{Showup}}(\cdot)}()$ \label{alg:ind:query-phase}
      \State $\epsilon^{l}_{0} = \EpoDic[\idzero]$, $\epsilon^{l}_{1} = \EpoDic[\idone]$
      \State \textbf{if} $\idzero, \idone \notin \mathcal{H}$ \textbf{then} return $\bot$ \label{alg:ind:both-honest}
      \If {
           $\quad\Ent[\idzero] \neq \Ent[\idone]$ \textbf{or} \\
           $\phantom{aaaa}|\{\Rev[\idzero], \Rev[\idone]\} \cap \BlockList^{*}| = 1$ \textbf{or} \\
           $\phantom{aaaa}(\epo^{*} > \min(\epsilon^{l}_{0}, \epsilon^{l}_{1})\allowbreak \land\allowbreak  \epo^{*} \leq \max(\epsilon^{l}_{0}, \epsilon^{l}_{1}))$} \label{alg:ind:no-trivial-wins}

          \State return $\bot$
      \EndIf
      \State $\ent^{b}, \tagH^{b}, \entproof^{b} \leftarrow \mathcal{O}_{\code{Showup}}(\idb, \epo^{*}, \BlockList^{*})$ \label{alg:ind:challenge-response}
      \State $b' \leftarrow \adv(\ent^{b}, \tagH^{b}, \entproof^{b})$
      \State return $b' = b$ \label{alg:ind:bit}
    \EndFunction
\end{algorithmic}
\end{algorithm}

Eventually the adversary outputs two honest challenge users $\idzero, \idone$, a challenge epoch $\epo^*$ and a challenge blocklist $\BlockList^{*}$ (line~\ref{alg:ind:query-phase}). The adversary's goal is to recognize one of these users during distribution. We rule out 3 trivial win conditions (line~\ref{alg:ind:no-trivial-wins}): (1) the selected recipients have different entitlements, (2) exactly one of the recipients is revoked, (3) exactly one of the tokens will run $\OnlyNameShowupT$ in the selected period $\epo^*$. If the adversary can determine the challenge bit $b$ (line~\ref{alg:ind:bit}), it wins the game. 

\begin{definition}
An aid-distribution system provides indistinguishability if the following advantage is negligible:
\[
    \code{Adv}_{\adv}^{\text{IND}} = \left|\prob\left[\expind[0](\secl)\right] - \prob\left[\expind[1](\secl)\right]\right|.
\]
\end{definition}

\begin{theorem}
  \label{thm:smart-card-indistinguishability}
The smart card system has indistinguishability provided that \PRFtwo\xspace is a pseudo-random function.
\end{theorem}

\begin{proof}[Proof of \cref{thm:smart-card-indistinguishability}]\label{proof:ind-card}
  For this, we consider the indistinguishability experiment $\expind$ in \Cref{alg:ind}. As a result of the check in line~\ref{alg:ind:both-honest}, we know that both users are controlled by the challenger. If the adversary triggers the conditions in line~\ref{alg:ind:no-trivial-wins}, it receives $\bot$, and hence, cannot win the game. We focus on the case where $\Ent[\idzero] = \Ent[\idone]$. Also, notice that if $\epo^{*} \leq \min(\epsilon^{l}_{0}, \epsilon^{l}_{1})$ (challenge epoch is before the last epoch of both cards) or $\{\Rev[\idzero], \Rev[\idone]\} \subset \BlockList^{*}$ (both cards blocked), then both cards $\idzero$ and $\idone$ will abort, and again the adversary cannot win the game. Therefore, going forward, we focus on the interesting case where $\epo^{*} > \max(\epsilon^{l}_{0}, \epsilon^{l}_{1})$ and $\Rev[\idzero], \Rev[\idone] \not\in \BlockList^{*}$. In this interesting case, both cards produce an output.

We use game-hopping to prove that the adversary has negligible advantage in winning the indistinguishability game. We start with the adversary running the indistinguishability experiment with $b$, i.e., $\expind$, and after a sequence of transitions, end up in the situation where the output is independent of the bit $b$. Our argument is that each of these steps is indistinguishable, and therefore the results follows.

The proof proceeds along the following sequence of games: 

\game{0} Let $G_0$ be the $\expind$ game. In this game, the adversary receives $\ent^{b}, \tagH^{b}, \entproof^{b}$ (see line~\ref{alg:ind:challenge-response}), where:
\begin{equation*}
  \tagH^b = \prf{\kHi{\id_b}}{\epox}
\end{equation*}

\game{1} Let $G_{1}$ be as $G_0$ but for all households $\id$, we replace $\PRFtwo(\kHi{\id}, \cdot)$ (used for generating the tags $\tagH$ in \ShowupT) with a uniformly random function $f_{\id}(\cdot)$ with matching domain and range.
Notice that this change affects oracle calls to $\OracleShowUp$ during the query phase (line~\ref{alg:ind:query-phase}) as well as how the challenge is computed (line~\ref{alg:ind:challenge-response}).

We argue that the change between $G_0$ and $G_{1}$ is indistinguishable for the adversary. Assume the adversary queries at most $n$ different households, with identifiers $\id_1, \ldots, \id_n$. We use an hybrid argument. Let $G_0^{i}$ be the game in which the PRFs of the first $i$ households $\id_1, \ldots, \id_i$ have been replaced by truly random functions. Then, $G_0^{0} = G_0$ and $G_0^{n} = G_1$.

To prove the indistinguishability of each step, we use the PRF assumption. We proceed by contradiction, and assume that there exists an adversary \adv that can distinguish $G_0^{i-1}$ from $G_0^{i}$ for some $i \in \{1, \ldots, n \}$. We then construct another adversary \advb that can break the PRF assumption.

Distinguisher \advb simulates the indistinguishability game for \adv. In
particular, \advb controls and simulates honest users. At the start of
the game \advb samples random functions $f_{\id_j}$ for all $1 \leq j < i$, as
well as random PRF keys $\kHi{\id_j}$ for all $i < j \leq n$. It answers oracle
calls to \OraclePrepareReg and \OracleFinishReg (for the honest users it
controls) normally.

We now show how \advb computes $\ShowupT{\stateT, \epo, \BlockList}$ for a household with identifier $\id_j$. Distinguisher $\advb$ only modifies how $\tagH$ is computed, and otherwise proceeds as normal. For $j < i$, \advb sets $\tagH = f_{\id_j}(\epo)$ and for $j > i$, \advb follows the original protocol and sets $\tagH = \prf{\kHi{\id_j}}{\epo}$. For the situation where $i = j$, \advb invokes the challenger in the PRF game on \epo to compute $\tagH$, i.e., if $b = 0$ in the PRF game, \advb obtains $\tagH = \prf{\kHi{\id_j}}{\epo}$; if $b = 1$, \advb obtains $\tagH^{j} = f^*(\epo)$ for a random function $f^*$. Given $\tagH$, \advb proceeds as normal to compute the commitment $\com$, the signature $\signatureAudit$, and the proof $\entproof$. At the end, \advb outputs whatever \adv outputs. Note that \advb perfectly simulates $G_0^{i-1}$ when $b = 0$, and perfectly simulates $G_0^{i}$ when $b = 1$. Therefore, because no adversary can break the PRF game, no adversary \adv can distinguish $G_0^{i-1}$ from $G_0^{i}$ with non-negligible probability.

To conclude the proof, we argue that the advantage of adversary \adv in game $G_1$ is negligible. In response to the challenge query it receives $(\ent^b, \tagH^b, \entproof^b)$ where
\begin{align*}
  \tagH^b &= f_{\id_b}(\epox) \\
  \com^b &= \code{Commit}(\ent^b, r), \quad r \randin \Zq \\
  \signatureAudit^b &= \code{Sign}(\sk, \tagH^b \parallel \epox \parallel \com^b \parallel \hBL), \\
  \entproof^b &= (\signatureAudit^b, \epox, \com^b, r),
\end{align*}
and $\hBL = H(\BlockList^*)$. Notice that $f_{\id_b}$ is a random function (for both $b=0$ and $b=1$) and that by construction, neither $f_{\idzero}$ nor $f_{\idone}$ was ever queried on $\epox$. Therefore, $\tagH^b$ is independent of $b$. Similarly, $\ent^b$ is the same for $b=0$ and $b=1$ and all households use the same signing key $\sk$. Hence, the response $(\ent^b, \tagH^b, \entproof^b)$ is independent of $b$. Thus, we conclude that the advantage of \adv is in winning $G_1$ is indeed negligible. The result in the theorem follows.
\end{proof}

\begin{theorem}
  \label{thm:smartphone-ind}
  The smartphone system has indistinguishability in the random oracle model provided that the DDH assumption holds, the ABC scheme is unlinkable, and the blocklisting scheme is anonymous.
\end{theorem}

\begin{proof}[Proof of \cref{thm:smartphone-ind}]
We proceed as in the proof of \cref{thm:smart-card-indistinguishability}. Since the logic of \OnlyNameShowupT is the same for the smart card and smartphone implementation, we know that we only need to consider cases where neither $\idzero$ nor $\idone$ is blacklisted, that the challenge epoch $\epox$ is fresh, and that the entitlements of $\idzero$ and $\idone$ are the same, i.e., $\Ent[\idzero] = \Ent[\idone]$.

We use game hopping to prove that the adversary has negligible advantage in winning the indistinguishability game. We start with the adversary running the indistinguishability experiment $\expind$ for the smartphone solution, and after a sequence of transitions, end up in the situation where the adversary's view is independent of the bit $b$. We argue that each of these steps is indistinguishable, and therefore the result follows.

The proof proceeds along the following sequence of games:

\game{0} Let $G_0$ be the $\expind$ game. In this game, the adversary receives $\ent^b, \tagH^b, \entproof^b$ (see line~\ref{alg:ind:challenge-response}) where
\begin{align*}
  \tagH^b &= H(\epo)^{\kHi{\id_b}},  \\
  \com  &= \code{Commit}(\ent^b, r), r \randin \Zq,
\end{align*}
the proof
\begin{multline*}
  \pi_s = \nizk\big\{ (C, \kH, \ent^b, \RevokeVal, r) : \tagH^b = H(\epo)^{\kH} \land\\
  C(\kH, \ent, \RevokeVal) \quad\land\quad \com = \code{Commit}(\ent, r) \\
  \land \quad \forall (h,H) \in \BlockList: H \not= h^{\RevokeVal} \big\},
\end{multline*}
and $\entproof^b = (\pi_s, \com, r)$.

\game{1} Let $G_1$ be as $G_0$, but instead we simulate the zero-knowledge proof $\pi_s$. The simulator takes as input the statement, which include the public data $\tagH, \com$ and the blocklist $\BlockList$, and it outputs a simulated proof. In our implementation, the proof is a non-interactive zero-knowledge proof obtained by applying the Fiat-Shamir heuristic to the Sigma protocol. These protocols can be simulated, assuming the random oracle for the hash function used in the Fiat-Shamir heuristic. As a result, the adversary cannot distinguish games $G_0$ and $G_1$.

\game{2} Let $G_2$ be as $G_1$, but in calls to \OnlyNamePrepare (i.e., the first step of the credential issuance process), the phone sets $\kH = 0$ so that the credential issuance process is independent of the household key. In later calls to \OnlyNameShowupT the challenger uses the real household key $\kH$. The adversary \adv, i.e., the registration station, cannot detect this change during the issuance process as the household key attribute $\kH$ is blinded and as per the changes in game $G_1$ the zero-knowledge proof is already simulated.

\game{3} Let $G_3$ be as $G_2$ but in the query phase (line~\ref{alg:ind:query-phase})
we no longer compute the tag $\tagH$ for household $\id$ and epoch $\epo$ as $\tagH = H(\epo)^{\kHi{\id}}$, but instead set $\tagH = \generator^{z_{\epo}^{\id}}$, where $z_{\epo}^{\id} \randin \Zq$. We do not yet modify how the challenge response (line~\ref{alg:ind:challenge-response}) is computed.

We argue that the change between $G_2$ and $G_{3}$ is indistinguishable for the adversary. Assume the adversary makes at most $n$ queries to the \OracleShowUp oracle.
We use a hybrid argument. Let $G_2^{i}$ be the game in which the tag in the first $i$ queries $\OracleShowUp(\id, \epo, \BlockList)$ has been replaced by $\tagH = \generator^{z_{\epo}^{\id}}$, whereas in the remaining queries the tag is still computed as $\tagH = H(\epo)^{\kHi{\id}}$. Then, $G_2^{0} = G_2$ and $G_2^{n} = G_3$.

We prove indistinguishability of each step. Assume by contradiction that there exists an adversary \adv that can distinguish $G_2^{i-1}$ from $G_2^{i}$ for $i \in \{1, \ldots, n\}$. We show that then there exists a distinguisher \advb that can break the DDH assumption.

Distinguisher \advb simulates the indistinguishability game for \adv by simulating honest users. It answers \OraclePrepareReg and \OracleFinishReg queries as before. At the start, \advb invokes its DDH oracle to obtain a DDH challenge $(X = \generator^x, Y = \generator^y, Z = \generator^z)$.

As part of the reduction, \advb programs the random oracle $H$. To do so, it guesses the epoch $\epo_i$ that will be asked in the $i$th query to \OracleShowUp. Assume without loss of generality that all queries to $H$ are unique. On input an epoch $\epo \neq \epo_i$, \advb picks $r_{\epo} \randin Zq$ and returns $\generator^{r_{\epo}}$ as the hash value. On input $\epo_i$ it returns $X$ (from the DDH challenge).

We focus on how \advb computes the tag $\tagH$ in for the $j$th $\OracleShowUp(\id, \epo, \BlockList)$ query. If $j < i$, \advb samples a fresh $z_{\epo}^{\id} \randin \Zq$ and sets $\tagH = g^{z_{\epo}^{\id}}$. To answer the query $j = i$, it proceeds as follows. If \advb guessed the wrong epoch $\epo_i$ above, it aborts. Otherwise, let $\id_i$ be the household queried.
It answers the $j$th query by returning $\tagH = Z$.
Going forward, $\advb$ will act as if the household key $\kHi{\id_i} = y$ despite not knowing $y$ (recall, $Y = \generator^y$ it is part of the DDH challenge).
Earlier \OracleShowUp queries for $\id_i$ used random tags, so knowing $\kHi{\id_i}$ was not needed. To answer a $j > i$ query to \OracleShowUp for household $\id_i$ and epoch $\epo$, \advb sets $\tagH = Y^{r_{\epo}}$. Recall that $H(\epo) = \generator^{r_{\epo}}$, therefore $\tagH = H(\epo)^{y} = H(\epo)^{\kHi{\id_i}}$ as required. For all other households $\id$, \advb simply uses the household key $\kHi{\id}$ to compute $\tagH = H(\epo)^{\kHi{\id}}$.

Observe now that if $z = xy$ in the DDH challenge, we have that for the $i$th query \advb sets $\tagH = Z = X^y = H(\epo_i)^{\kHi{\id_i}}$ as in game $G_2^{i - 1}$. And if $z \randin \Zq$ we have that $\tagH = g^z$ for $z \randin \Zq$ as in game $G_2^{i}$.
Therefore, any advantage that \adv has in distinguishing $G_2^{i - 1}$ from $G_2^{i}$ results in \advb being able to break the DDH assumption. We conclude that $G_2$ and $G_3$ are indistinguishable.

\game{4} As in game $G_3$, but now we also replace the tag $\tagH$ returned in the challenge response by a random value. The proof is exactly the same as in the previous step, but now requires only a single hop.

To conclude the proof, we argue that the advantage of the adversary in game $G_4$ is negligible. Recall that \adv receives $(\ent^b, \tagH^b, \entproof^b=(\pi_s^b, \com^b, r^b))$, we argue all of these values are independent of the challenge bit $b$. By design, $\ent^0 = \ent^1$ and therefore $\ent^b$ as well as the commitment $\com^b$ with randomness $r^b$ are independent of $b$. As per the changes in game $G_1$, the proof $\pi_s^b$ is independent of $b$ as well. Finally, as per the changes in games $G_2$ through $G_4$, the tag $\tagH^b$ is  random group element, and therefore also independent of $b$. The result follows.
\end{proof}

\subsection{Auditability}\label{subsec:auditability}
Aid-distribution systems should be auditable (\reqlink{auditS}): a malicious distribution station should not be able to convince an auditor that it provided more aid than what corresponds to the legitimate recipients that showed up. Since a malicious registration station could simply enroll more recipients as legitimate, making auditing moot, we assume the registration station is honest. We model auditability using the experiment in Algorithm~\ref{alg:aud}.
The adversary can register honest recipients using $\OracleHonestReg$. The oracle returns recipients' revocation value to enable blocklisting of those recipients.
To register malicious recipients, the adversary uses $\OracleMalUserReg$.
We model the fact that smart cards are TEEs by making the $\OracleMalUserReg$ oracle unavailable to smart card attackers.

\begin{algorithm}[tbp]
  \caption{Auditability experiment}\label{alg:aud}
  \begin{algorithmic}[1]
    \small
    \Function{$\expaud$}{$\secl$}
      \State $\param \leftarrow \GlobalSetup{\secl}$
      \State $\skrs, \pkrs \leftarrow \SetupRS{\secl}$
      \State $\epo^{*}, \entsum^{*}, \AuditProof^{*}, \BlockList^{*} \leftarrow \adv^{\OracleHonestReg(\cdot),\allowbreak \mathcal{O}_{\code{MalUserReg}}(\cdot),\allowbreak \mathcal{O}_{\code{Showup}}(\cdot)}(\pkrs)$\label{alg:aud:forgery}
      \State $\code{valid} \gets \VerifyA{\pkrs, \epo^{*}, \entsum^{*}, \AuditProof^{*},\allowbreak \BlockList^{*}}$ \label{alg:aud:check-proof}
      \State $\mathcal{R}_{\text{mal}} = \{ \id \mid \Rev[\id] \in \BlockList^{*} \land \id \in \mathcal{M} \}$ \label{alg:aud:revoked}
      \State $\entmax \gets \entsum[\epo^{*},\BlockList^{*}] + \sum_{\id \in \mathcal{M} \setminus \mathcal{R}_{\text{mal}}} \Ent[\id]$ \label{alg:aud:max-ent}
      \State return $\code{valid} \land (\entsum^{*} > \code{ent}_{max})$
    \EndFunction
  \end{algorithmic}
\end{algorithm}

At the end, the adversary produces a target epoch $\epo^*$, entitlement $\entsum^{*}$, proof $\AuditProof^{*}$, and blocklist $\BlockList^{*}$ (line~\ref{alg:aud:forgery}). The game checks if the produced proof is valid (line~\ref{alg:aud:check-proof}) and determines the maximum explainable entitlement $\entmax$. The maximum consists of two terms (line~\ref{alg:aud:max-ent}): (1) the non-revoked honest users for which the adversary called $\OracleShowUp$ for epoch $\epox$ and blocklist $\BlockList^*$, and (2) the non-revoked malicious users the adversary controls.
The adversary wins if the proof verifies, and the total entitlement exceeds $\entmax$.

\begin{definition}
An aid-distribution system is auditable if the following probability is negligible:
\[
    \code{Succ}^{\text{AUD}}(\adv) = \prob\left[\expaud(\secl) = 1 \right].
\]
\end{definition}

\begin{theorem}
  \label{theo:card-auditability}
  The smart card system is auditable provided that
  the digital signature scheme is unforgeable and
  the discrete logarithm assumption holds in $\group$.
\end{theorem}

\begin{proof}[Proof of \cref{theo:card-auditability}]
  We show that no adversary \adv can win the \expaud game in \Cref{alg:aud}. Recall that because each smart card is a TEE, the adversary is not allowed to call $\OracleMalUserReg$. Recall that the audit proof $\AuditProof^*$ produced by the adversary is of the form
  \begin{equation*}
    \AuditProof^* = (\rsum^{*}, (\signatureAudit^{(i)}, \tagH^{(i)}, \com^{(i)})_i, \hBL^{*}),
  \end{equation*}
  and the audit check on $\AuditProof^*$ for total entitlement $\entsum^*$, epoch $\epox$ and blocklist $\BlockList^*$ passes.
  
  At a high level, the adversary has two ways to win the game: (1) by forging a
  new record with a valid signature that can pass the verification of an honest
  auditor; and (2) by finding another opening to the product of Pedersen commitments. We prove that neither of the two is possible.
  
  First, we argue that each signature $\signatureAudit^{(i)}$ must have been produced by an honest token (e.g., using a call to $\OracleShowUp$). Suppose not, then \adv must break the unforgeability of the signature scheme. More formally, suppose one of the signatures was not produced by an honest token, then we construct an adversary \advb against the unforgeability of the signature scheme. Adversary \advb receives as input the challenge verification key $\pk$ which it passes to \adv in line~\cref{alg:aud:forgery} of the \expaud game. For every call by adversary \adv to $\OracleShowUp$ the signature-adversary \advb follows the protocol, except that it uses its signing oracle to compute the signature. At the end, \adv outputs an audit proof $\AuditProof^*$. If at least one signature was not produced by an honest token in the original game, then $\AuditProof^*$ contains a forged signature (on a new message because all $\tagH$s are different). Adversary \advb outputs this signature to win the unforgeability game.

  We conclude that all signatures were produced by honest tokens. 
Honest users whose revocation value $\RevokeId$ appear on the blocklist $\BlockList^*$ do not produce any signatures (and these thus cannot show up in the forged audit proof). Unblocked honest tokens, do produce a signature, but will do so only once for each epoch, including the target epoch $\epox$. And these signatures contain commitments $\com^{(i)}$ to the actual entitlements of these honest users.

  Suppose now that \adv wins the $\expaud$ game without forging signatures. Because each commitment $\com^{(i)}$ in $\AuditProof^*$ must have been produced for epoch $\epox$ by an honest unrevoked user, we know that $\com = \prod_i \com^{(i)}$ must commit to a value $\entsum$ that is less or equal to $\entmax$. Let $\rreal$ such that $\com = \comcommit(\entsum, \rreal)$. Since \adv wins the game, it must output a values $\entsum^*$ and $\rsum^*$ such that $\com = \comcommit(\entsum^*, \rsum^*)$ and $\entsum^* > \entmax$. Therefore $\entsum \neq \entsum^*$, so the adversary has produced two openings of commitment $\com$. However, Pedersen commitments are binding assuming the discrete logarithms assumption, therefore such an adversary cannot exist.
\end{proof}

\begin{theorem}
  \label{thm:aud:smartphone}
  The smartphone system is auditable provided that
  the ABC scheme is unforgeable and the discrete logarithm holds in $\group_1$.
\end{theorem}

\begin{proof}[Proof of \cref{thm:aud:smartphone}]
  The auditability proof for the smartphone system proceeds similarly to the proof for the smart card system: to win, the adversary must either forge an attribute-based credential, or it must break the binding property of Pedersen commitments. Neither is possible.

  More formally, we show that no adversary \adv can win the \expaud game in \Cref{alg:aud}. Contrary to the smart card solution, the adversary is allowed to call $\OracleMalUserReg$. The audit proof $\AuditProof^*$ produced by the adversary is of the form:
\begin{equation*}
    \AuditProof^* = (\rsum^{*}, (\pi_s^{(i)}, \tagH^{(i)}, \com^{(i)})_i).
  \end{equation*}
  and the audit check on $\AuditProof^*$ for total entitlement $\entsum^*$, epoch $\epox$ and blocklist $\BlockList^*$ passes.
  
  We first focus on the zero-knowledge proofs
  \begin{multline*}
    \pi_s^{(i)} = \nizk\big\{ (C, \kH, \ent, \RevokeVal, r) : \tagH^{(i)} = H(\epox)^{\kH} \land\\
    C(\kH, \ent, \RevokeVal) \,\land\, \com^{(i)} = \code{Commit}(\ent, r) \\
    \land\, \forall (h,H) \in \BlockList^*: H \not= h^{\RevokeVal} \big\}.
  \end{multline*}
  We argue that $\pi_s^{(i)}$
  must be of one of two types: (1) the proof contains a credential showing proof for a credential that was issued to an honest user, or (2) it contains a showing proof for a credential that was issued using $\OracleMalUserReg$ to a corrupted user. Suppose to the contrary that $\pi_s^{(i)}$
  contains a showing proof for a different credential. Then, \adv must have produced a forgery for the underlying ABC scheme, which we assumed is not possible.
  
  So, each proof $\pi_s^{(i)}$
  must be over a credential that was issued by the registration station (controlled by the challenger). 
Because of the soundness of the zero-knowledge proofs, we therefore know that all $\com^{(i)}$ must commit to either entitlements of honest households $\id$ for which $\OracleShowUp$ was called in epoch $\epox$ or to a malicious household. Moreover, the blocklist part of the proof ensures that none of these households were on the blocklist $\BlockList^*$. And finally, the fact that all tags $\tagH^{(i)}$ are different ensures that no single household's credential is repeated.
  
We now show that if the adversary can win the \expaud game without forging credentials, then it must break the binding property of the commitment scheme. Let $\com = \prod_i \com^{(i)}$ be the combined commitment. Because the audit proof verifies, we know that $\com = \comcommit(\entsum^*, \rsum^{*})$. We now construct another opening for $\com$. First, for all honest users, we know the values $\ent^{(i)}, r^{(i)}$ such that $\com^{(i)} = \comcommit(\ent^{(i)}, r^{(i)})$. For the malicious users, we can find these openings by applying the knowledge extractor to the zero-knowledge proofs. (In our case, we use non-interactive zero-knowledge proofs in the random oracle model, the witnesses can therefore be extracted by rewinding the adversary.) Let $\ent = \sum_i \ent^{(i)}$ and $\rreal = \sum_i r^{(i)}$, then $\com = \comcommit(\entsum, \rreal)$ by construction. However, we know that each of these commitments was produced by a registered (honest or malicious) user that was not revoked, hence $\entsum \leq \entmax$, while $\entsum^{*}> \entmax$, so $\entsum \not= \entsum^{*}$. Therefore we have constructed two openings for $\com$ which violates the binding property of the commitment scheme. However, Pedersen commitments are binding assuming the discrete logarithm assumption holds, therefore such an adversary cannot exist.
  \end{proof}

\subsection{Entitlement Privacy Game}
An aid-distribution system should not unnecessarily disclose sensitive information of recipients to auditors (\reqlink{auditP}). By design, the auditor learns the sum of entitlements and the number of legitimate recipients. We model that the auditor does not learn \emph{more} using a two-world game inspired by \emph{Benaloh’s ballot privacy game}~\cite{BernhardCGPW15}, see Algorithm~\ref{alg:ent}.

In this game, the registration station and distribution station are honest. As before, the adversary can request the creation of honest recipients using $\OracleHonestReg$. The key difference is that the game models \emph{two} (parallel) distribution stations. The adversary can request that honest recipients show up to these stations using the oracle $\OracleShowUpTwo$, which takes as input two recipient identifiers, who interact with the respective distribution stations. The $\OracleShowUpTwo$ oracle tracks which blocklist the adversary uses in each epoch.

In the entitlement privacy experiment, the adversary outputs a challenge epoch $\epo^*$ (line~\ref{alg:ent:target}). The two distribution stations will produce their audit proofs for $\epo^*$ (lines~\ref{alg:ent:proof-one}--\ref{alg:ent:proof-two}). Recall that the auditor always learns the total entitlement, the number of recipients and the blocklist used. Therefore, the adversary loses if these are not the same in the two worlds (lines~\ref{alg:ent:worlds-same}--\ref{alg:ent:world-fail}). 
(Note that requiring that $\BlockList$ is the same in each call to $\OracleShowUpTwo$ is not sufficient to prevent trivial wins because a call to \OnlyNameShowupT could fail in only one of the two worlds. This difference would be detectable.)
Finally, the auditor receives one of the proofs and needs to guess which of the two it is (lines~\ref{alg:ent:guess}--\ref{alg:ent:win}). Any extra leakage will let the adversary win the game.

\begin{algorithm}[tbp]
    \caption{Entitlement privacy experiment}\label{alg:ent}
    \begin{algorithmic}[1]
      \small
      \Function{$\expent$}{$\secl$}
        \State $\param \leftarrow \GlobalSetup{\secl}$
        \State $\skrs, \pkrs, \leftarrow \SetupRS{\secl}$
\State $\epo^{*} \leftarrow \adv^{\mathcal{O}_{\code{HonestReg}}(\cdot), \mathcal{O}_{\code{ShowupTwo}}(\cdot)}(\pkrs)$ \label{alg:ent:target}
        \State $\entsum^0, \AuditProof^0 \leftarrow \CreateAuditProofDS{\epo^{*}, \transcript^0, \BlockLists[\epo]^{*}}$ \label{alg:ent:proof-one}
        \State $\entsum^1, \AuditProof^1 \leftarrow \CreateAuditProofDS{\epo^{*}, \transcript^1, \BlockLists[\epo]^{*}}$ \label{alg:ent:proof-two}
        \If {$\entsum^0 \neq \entsum^1$ \textbf{or} $|\AuditProof^0| \neq |\AuditProof^1|$ \textbf{or}\\
             \phantom{aaaa}$\lvert\BlockLists[\epox]\rvert > 1 $} \label{alg:ent:worlds-same}
            \State return $\bot$ \label{alg:ent:world-fail}
        \EndIf
        \State $b' \leftarrow \adv(\AuditProof^b)$ \label{alg:ent:guess}
        \State return $b' = b$ \label{alg:ent:win}
      \EndFunction
    \end{algorithmic}
\end{algorithm}

\begin{definition}
Aid-distribution systems provide entitlement privacy if the following advantage is negligible:
\[
    \code{Adv}_{\adv}^{\text{ENT}} = \left|\prob\left[\expent[0](\secl)\right] - \prob\left[\expent[1](\secl)\right]\right|.
\]
\end{definition}

\begin{theorem}\label{thm:card_privacy_ent}
  The smart card system maintains entitlement privacy against a malicious auditor assuming the commitment scheme is hiding.
\end{theorem}

\begin{proof}[Proof of \cref{thm:card_privacy_ent}]
Consider the entitlement privacy experiment \expent in \Cref{alg:ent}. By making calls to the \OracleShowUpTwo oracle, the adversary constructs two parallel worlds. As a result of the check in line~\ref{alg:ent:worlds-same} we know that the total entitlements for the challenge epoch $\epox$ must be the same in both worlds, i.e., $\entsum := \entsum^0 = \entsum^1$ and that the audit proofs contain the same number of records.

Recall that the audit proof is given by
\begin{equation*}
\AuditProof^b = (\rsum^{b}, (\signatureAudit^{(b,i)}, \tagH^{(b,i)}, \com^{(b,i)})_i, \hBL^b).
\end{equation*}
By construction, we must have that
\begin{equation}
\label{eq:commitment-sum}
\comcommit(\entsum, \rsum^b) = \prod_{i=1}^{n_{\epo}} \com^{(b,i)}
\end{equation}
and $\hBL := \hBL^0 = \hBL^1$. We proceed by game-hopping to show that the adversary cannot use the audit proof $\AuditProof^b$ to determine the value $b$. When the changes in these game hops are viewed together, they simulate an audit proof $\AuditProof$ using just the total entitlement.

\game{0} Game $G_0$ is as \expent.

\game{1} Game $G_1$ is as game $G_0$, except that the challenger directly computes the signatures $\signatureAudit^{(b,i)}$ (using the honest users' signing key $\sk$) as
\begin{equation*}
\signatureAudit^{(b,i)} = \code{Sign}(\sk, \tagH^{(b,i)} \parallel \epox \parallel \com^{(b,i)} \parallel \hBL).
\end{equation*}
The challenger computes the signatures in exactly the same way as the tokens, therefore this change is not detectable by the adversary.

\game{2} Game $G_2$ is as game $G_1$ but now the challenger sets
\begin{equation*}
\com^{(b,1)} = \comcommit(\entsum, \rsum^b) / \prod_{i=2}^{n_{\epo}} \com^{(b,i)}
\end{equation*}
without modifying $\entsum$ or $r_{sum}^b$. This change is syntactic, therefore the distributions of $G_2$ and $G_1$ are identical.

\game{3} Game $G_3$ is as game $G_2$ except that the challenger replaces the commitments $\com^{(b,i)}$ for $i \geq 2$ with fresh commitments to 0. Note that by the change in game $G_2$ equation \eqref{eq:commitment-sum} still holds. We now show using a hybrid argument that \adv cannot detect this change.

Let $G_2^{i}$ be the game where commitments $\com^{(b,2)}, \ldots, \com^{(b,i)}$ have been replaced by fresh commitments to 0. Then $G_2 = G_2^{1}$ and $G_3 = G_2^{n}$. Assume by contradiction that \adv can distinguish $G_2^{i - 1}$ from $G_2^{i}$ for $i \in \{2, \ldots, n\}$. Note that in $G_2^{i - 1}$ we set
\begin{equation*}
\com^{(b, i)} = \comcommit(\ent^{(i)}, r^{(i)})
\end{equation*}
whereas in $G_2^{i}$ we set
\begin{equation*}
\com^{(b, i)} = \comcommit(0, r_i').
\end{equation*}
We make no further changes. Any adversary that distinguishes these two games must break the hiding property of the commitment scheme. Since the Pedersen commitments that we use are information theoretically hiding, we conclude that no adversary \adv can distinguish $G_2^{i - 1}$ from $G_2^{i}$. By the hybrid argument, no adversary can distinguish $G_3$ from $G_2$.

\game{4} Game $G_4$ is as game $G_3$ except that the challenger no longer computes tag $\tagH^{(b,i)} = \tagH = \prf{\kHi{i}}{\epo}$ but instead picks a fresh household key $\kHi{i}'$ and sets $\tagH^{(b,i)} = \prf{\kHi{i}'}{\epo}$. Note that this change is purely syntactic:
(1) by construction households are unique, and
(2) in the \expent game the adversary receives no other values computed using the household keys.

In game $G_4$, the audit proof $\AuditProof^b$ sent to the adversary is constructed independently of the bit $b$. Therefore, the result follows.
\end{proof}

\begin{theorem}\label{thm:phone_privacy_ent}
  The smartphone system maintains entitlement privacy against a malicious auditor assuming that the commitment scheme is hiding and the ABC disclosure proofs are simulatable in the random oracle model.
\end{theorem}
\begin{proof}[Proof of \cref{thm:phone_privacy_ent}]
This proof is similar to the entitlement privacy proof for the 
card solution. As before, we have that $\entsum := \entsum^0 = \entsum^1$; the audit proofs contain the same number of records; and
all individual zero-knowledge proofs are with respect to the same blocklist $\BlockList$.
In the smartphone-based solution, the auditor gets $(\entsum, \AuditProof)$ where 
\begin{equation*}
\AuditProof^b = (r_{sum}^{b}, (\pi_s^{(b,i)}, \tagH^{(b,i)}, \com^{(b,i)})_i).
\end{equation*}
By construction, we must have that \eqref{eq:commitment-sum} holds here as well.

We proceed by game-hopping to show that the adversary cannot use the audit proof $\AuditProof^b$ to determine the value $b$. As in the phone proof, we effectively construct a simulator for $\AuditProof^b$ that relies only on $\entsum$.

\game{0} Game $G_0$ is as \expent.

\game{1} Game $G_1$ is as game $G_0$ but instead the challenger simulates all zero-knowledge proofs
$\pi_s^{(b,i)}$ using only public information: the tag $\tagH^{(b,i)}$, the blocklist $\BlockList$, and the commitment $\com^{(b,i)}$. The adversary cannot detect this change (see $G_1$ in the proof of \Cref{thm:smartphone-ind}).

\game{2} We repeat the same change as in the previous proof for the smart phone solution. Game $G_2$ is as $G_1$ but now the challenger sets
\begin{equation*}
\com^{(b,1)} = \comcommit(\entsum, \rsum^b) / \prod_{i=2}^{n_{\epo}} \com^{(b,i)}
\end{equation*}
without modifying $\entsum$ or $\rsum^b$. This change is syntactic, therefore the distributions of $G_2$ and $G_1$ are identical.

\game{3} Game $G_3$ is as $G_2$ except that for $i > 1$ the challenger replaces the commitment $\com^{(b, i)}$ by a fresh commitment to zero. The same argument as for the smart card solution (using the hiding property of the commitment scheme) shows that the adversary cannot detect this change.

\game{4} Game $G_4$ is as $G_3$ except that challenger uses a fresh household key $\kHi{i}'$ to compute the tag $\tagH^{(b,i)} = H(\epox)^{\kHi{i}'}$. Also in the smart phone solution, no other values depend on the household keys (the registration and distribution station are both honest and do not collude with the adversary). Therefore, this is a syntactic change.

In game $G_4$, the audit proof $\AuditProof^b$ sent to the adversary is constructed independently of the bit $b$. Therefore, the result follows.
\end{proof}

\subsection{Security of Aid Distribution}

Aid-distribution systems should ensure that illegitimate recipients cannot receive aid (\reqlink{legitimate}), and that legitimate recipients cannot receive more aid than entitled (\reqlink{limit}). The security game, the theorems and proofs are very similar to the auditability proofs (Sect.~\ref{subsec:auditability}), but with more restrictions on the adversary's power. In the security game, the adversary controls only users, and must convince an \emph{honest} distribution station to hand out more aid than permitted. Because the adversary no longer controls the distribution station, we give it access to an explicit $\OracleVerifyEnt$ oracle in the security game in Algorithm~\ref{alg:sec} through which it can simulate users (honest or malicious) showing up at the distribution station.

The adversary can, as before, create malicious users (using $\OracleMalUserReg$) as well as honest users (using $\OracleHonestReg$). It can use \OracleShowUp to obtain show up data of honest users, which it could then feed into $\OracleVerifyEnt$. See line~\ref{alg:sec:query}. We model the fact that smart cards are TEEs, by making the $\OracleMalUserReg$ and $\OracleShowUp$ oracles 
unavailable to smart card attackers.

The oracle \OracleVerifyEnt tracks show-up events that the distribution station would accept in $\transcript$. The adversary's goal is to convince the distribution station to, for a given round and blocklist, to hand out more aid than the joint entitlement of all non-revoked malicious users controlled by the adversary \emph{and} non-revoked honest users that showed up. This model then covers: illegitimate recipients receiving aid; legitimate recipients receiving more aid than entitled; a group of colluding users receiving more aid than they are jointly entitled to; and revoked users receiving aid.

After interacting with oracles, the adversary outputs a target epoch $\epo^*$ and a blocklist $\BlockList^*$. First, the challenger computes the total amount of aid handed out in terms of successful calls to $\OracleVerifyEnt$ for this choice of blocklist (line~\ref{alg:sec:seen}). Then, as in the indistinguishability game, we compute the maximum allowable entitlement $\entmax$ in two parts (see line~\ref{alg:sec:max}). The $\entsum$ term captures the aid requested by non-revoked honest users using $\OracleShowUp$. The second term accounts for the non-revoked malicious users (see line~\ref{alg:sec:revoked}). If the accepted aid is larger than $\entmax$, the adversary wins.

\begin{algorithm}[tbp]
  \caption{Security experiment}\label{alg:sec}
  \begin{algorithmic}[1]
    \small
    \Function{\OracleVerifyEnt}{$\epo, \ent, \tagH, \entproof, \BlockList$}
      \If {$\VerifyDS{\pk, \epo, (\ent, \tagH,\entproof),\allowbreak \BlockList} \land \tagH \not\in \transcript$}
        \State $\transcript[\epo] \leftarrow \transcript[\epo] \cup \{(\ent, \tagH, \entproof, \BlockList)\}$ \label{alg:sec:logs}
      \EndIf
    \EndFunction
    \vspace{2mm}

    \Function{$\expsec$}{$\secl$}
      \State $\param \leftarrow \GlobalSetup{\secl}$
     \State $\sk, \pk, \leftarrow \SetupRS{\secl}$
      \State $\epo^{*}, \BlockList^* \leftarrow \adv^{\OracleHonestReg(\cdot),\allowbreak \OracleMalUserReg(\cdot),\allowbreak \OracleShowUp(\cdot), \OracleVerifyEnt(\cdot)}(\pk)$ \label{alg:sec:query}
      \State $\entseen \gets \sum \{ \ent \mid (\ent, \cdot, \cdot, \BlockList^*) \in \transcript[\epo^*] \}$ \label{alg:sec:seen}
      \State $\mathcal{R}_{\text{mal}} = \{ \id \mid \Rev[\id] \in \BlockList^{*} \land \id \in \mathcal{M} \} $ \label{alg:sec:revoked}
      \State $\entmax \gets \entsum[\epo^{*},\BlockList^{*}] + \sum_{\id \in \mathcal{M} \setminus \mathcal{R}_{\text{mal}}} \Ent[\id]$ \label{alg:sec:max}
      \State return $\entseen > \entmax$
    \EndFunction
  \end{algorithmic}
\end{algorithm}

\begin{definition}
An aid-distribution system is secure if the following probability is negligible:
\[
    \code{Succ}^{\text{SEC}}(\adv) = \prob\left[\expsec(\secl) = 1 \right].
\]
\end{definition}

We prove the security of distribution in both solutions using a reduction to auditability. 
\begin{theorem}\label{thm:sec_distribution}
  If a scheme is auditable, this scheme also provides security of distribution.
\end{theorem}  
\begin{proof}[Proof of \cref{thm:sec_distribution}]
  Assuming there exists an adversary \adv that breaks the security of distribution, i.e., \adv can win $\expsec$ in \cref{alg:sec} with non-negligible advantage, we construct an adversary \advb that breaks auditability by winning $\expaud$ in \cref{alg:aud} with at least the same advantage. 

  At the start of the auditability game, \advb receives \param and the public key $\pk$. It relays these to the security adversary \adv. Whenever \adv calls the three oracles, i.e., \OracleHonestReg, \OracleMalUserReg, and \OracleShowUp, \advb will relay the query to its challenger and send the response back to \adv. Whenever \advb receives a call to \OracleVerifyEnt from \adv, it executes this oracle as stated, and stores the result in $\transcript[\epo]$.

After the oracle queries, \adv outputs a chosen period $\epo^*$ and a target blocklist $\BlockList^*$. Auditability adversary \advb uses all recorded responses in $\transcript[\epo^*]$ to compute $\entsum^*, \AuditProof^*\leftarrow \CreateAuditProofDS{\epo,\allowbreak \transcript[\epo^*]}$ and finally outputs $(\epox, \entsum^*, \AuditProof^*, \BlockList^*)$. Assuming this audit proof is valid, \advb wins the audit game because $\entseen > \entmax$, and hence, $\entsum^{*} > \entmax$.

Finally, we show that $\AuditProof^*$ is valid. The verification oracle \OracleVerifyEnt and the verification function \OnlyNameVerifyA running by auditors do the same check on the output tuple of \OnlyNameShowupT functions. Hence, the winning condition for \adv and for \advb are essentially the same.
\end{proof}

\section{Performance Evaluation}
In this section, we evaluate the suitability of the instantiations in Sections~\ref{sec:card} and~\ref{sec:phone} for real deployment. We focus on the performance of the registration and the distribution phases. 
In particular, we focus on the computation and the communication operations that happen on the token and may affect user experience. 
We omit the setup and audit phases, because they can be done offline, and thus, are not a bottleneck.
Our code can be found here: \url{https://github.com/spring-epfl/not-yet-another-id-code}.

\vspace{1.5mm}\parBF{Smart-card-based Solution}
We implemented a prototype of the smart-card-based solution on a \emph{NXP J3H145 dual 144k} Java Card~\cite{JavaCardUserGuide21}.
We focus on the functions running on the smart card at registration and distribution, as the station can run on powerful hardware. The measurements include communication cost between card and reader. We report the mean over at least 5 runs. In all cases, the standard error of mean (SEM) is below 1\%.

At registration, the card receives $(\ent,\allowbreak \RevokeId)$, of 32 bytes each. Then, it runs \OnlyNameFinish and generates a 32-byte secret $\kH$. 
In total, computation and communication at registration runs in under 100\,ms.

We determine the cost of running \OnlyNameShowupT by measuring the cost of the individual operations and the transfer cost.
First, the card downloads the period $\epo$ (8 bytes) and the blocklist \BlockList. Assuming 512 entries of 32 bytes, transferring, checking, and hashing the blocklist take 3.3\,s. The time scales linearly in the length of the blocklist. 
Second, the card updates the counter \lastepoch and computes a 32-byte tag $\tagH$. We implement the PRF for computing $\tagH$ using AES in ECB mode. This step takes 113\,ms.
Third, the card computes the Pedersen commitment on the entitlement. Since the NXP J3H145 dual 144k card does not expose a direct elliptic-curve API, we use JCMathLib~\cite{MavroudisS20} and the card's Diffie-Hellman key exchange functionality to compute the commitment. Computing $\com$ takes 500\,ms. Finally, signing takes 166\,ms. Adding these together (except the first one), we estimate the baseline cost of $\OnlyNameShowupT$ to be around 779\,ms.

With 512 entries on the blocklist, the distribution protocol takes less than 5 seconds in total. This running time is smaller than any physical interaction happening at distribution time, and hence our protocol can fulfill \reqlink{scale}.

\vspace{1.5mm}\parBF{Smartphone-based Solution}
We implemented a Rust~\cite{Rust22} prototype of the smartphone solution. 
Our prototype includes the Pedersen Commitment scheme~\cite{TorbenNoninteractive91} and the Pointcheval-Sanders scheme~\cite{PointchevalS16} using SHA-256 as the hash function. 
In our experiments we use a \emph{Samsung Galaxy a40} smartphone with a \emph{Samsung Exynos 7 Octa 7904} chipset for recipients, and a computer with an \emph{Intel(R) Core(TM) i7-7600U CPU 2.80GHz} as the stations.
We report running times (averaged over 16 runs) and transfer costs. Transfer times will depend on the communication channel. For example, Bluetooth can achieve 1\,Mb/s. Appendix~\ref{subsec:channel} discusses other channels.

At registration the phone runs \OnlyNamePrepare and \OnlyNameFinish which take 4.6\,ms and 0.6\,ms respectively. The registration station runs \OnlyNameProcess which takes 0.8\,ms. The request sent from the phone to the station and the response from the station to the phone are both less than a few hundred bytes.

\begin{figure}[tbp]
    \centering
    \includegraphics[width=\columnwidth]{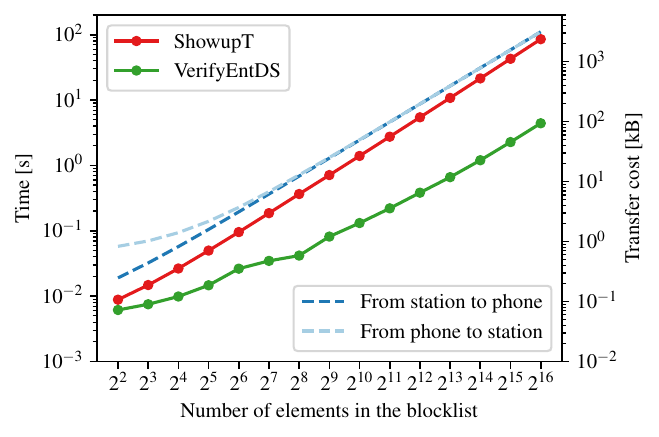}
    \caption{Smartphone solution: Performance at distribution. The computation cost of \OnlyNameShowupT (\textcolor{showupt}{red}), computation cost of \OnlyNameVerifyDS (\textcolor{verifyentds}{green}), the transfer cost of station$\rightarrow$phone (\textcolor{stationtophone}{blue}), and the transfer cost of phone$\rightarrow$station (\textcolor{phonetostation}{light blue}). All are linear in the length of the blocklist.}
    \label{fig:eval}
\end{figure}

\Cref{fig:eval} shows the computation cost of running \OnlyNameShowupT on the phone, the computation cost of verifying the proof (\OnlyNameVerifyDS) on the station, as well as the communication costs.
Even with 1024 blocked households, the whole distribution protocol can still finish within seconds, making this solution efficient enough for supporting aid distribution at scale (\reqlink{scale}).

\section{Practical Considerations}
\label{sec:practical}

\parBF{Multiple Registration and Distribution Stations}
While our design focuses on cases where there is only one registration and distribution station, humanitarian organizations may need to set up multiple registration and distribution sites for large-scale operations~\cite{Icrcecosec20}.  
For example, having multiple sites reduces the waiting time and therefore lowers the risk of attack for staff and recipients~\cite{Icrcecosec20}.
But having multiple sites can also affect the security and privacy offered of the system. 

When there are multiple registration stations, recipients from the same household can register several times and get goods multiple times per period. 
To avoid this problem, registration stations need access to the validation means of other registration stations (e.g., communication with the elders of another village).

When there are multiple distribution sites, recipients can ask for goods in more than one of them.
To avoid this problem, distribution stations need to synchronize their lists of seen tags $\tagH$ to learn which households already requested goods in a period. Synchronization requires communication, e.g., Internet connectivity or transportation of digital copies (e.g., in a USB stick) from one site to the other.
Communication may not be possible, or be difficult in some emergency scenarios where aid distribution takes place.

In absence of connectivity, assigning each household to a fixed station would prevent multiple-station-based double dipping at the cost of flexibility for recipients.

\parBF{Preventing Delegation of Tokens}
\label{subsec:auth}
One of the problems that the \icrcAbbr staff finds in the field is that recipients may willingly or unwillingly, give their entitlement proofs to others.
When this happens, illegitimate recipients gain access to goods. 
To avoid this issue, the distribution station needs to authenticate recipients to ensure that the recipient is indeed registered at the registration station.

The most popular authentication methods are not suitable in our case. 
For example, (graphical) passwords or tokens can be delegated and thus do not solve the problem.
Due to the ease of use and the inherent protection against delegation, biometrics are perceived by the humanitarian sector as a very advantageous solution.

Biometric-based authentication, while desirable, also raises concerns.
Biometrics contain private information about individuals, and they are not renewable. 
Once they are leaked or shared, those having the biometric data can use them to re-identify their owners forever.
Humanitarian organizations are already under pressure from state and non-state actors to share or disclose biometric data they have~\cite{HayesICRC19}. 
This is why, to reduce risks for recipients, the \icrcAbbr has a strict policy for the usage of biometrics~\cite{Icrcbiopolicy19} which recommends avoiding the creation of biometrics databases.
Such biometrics policy is in line with other relevant legal frameworks under which humanitarian organizations may have to operate, e.g., the EU General Data Protection Regulation~\cite{GDPR} or the African Union Convention on Cybersecurity and Personal Data Protection~\cite{AUdataprotection14}.

Our design can be easily adapted to support biometric authentication while respecting the \icrcAbbr policies.
It suffices to store the biometrics on the tokens at registration and implementing the biometric authentication \emph{inside} the token.
At the distribution station, the user can use the token to prove their identity.
To keep the system properties, it is necessary that the biometric sensors and devices realizing the computation (smart cards or smartphones) are trusted.

\parBF{Limits of technological solutions against delegation}
While biometrics (and other strong authentication mechanisms) can prevent delegation, they cannot fully prevent illegitimate recipients from getting access to aid. 
For example, recipients may be coerced to give out their aid. This problem is hard to solve with technology, as digital systems can check the authenticity or eligibility of the recipient, but not their willingness.

\section{Conclusion}

Humanitarian organizations see digitalization as an opportunity to increase their operations' efficiency and therefore increase their capability to help people in need.
However, they also recognize that the introduction of technology may bring new risks for the populations they serve\cite{KaspersenICRC16}.

In this paper we tackled the case of digitalization of aid-distribution programs.
In collaboration with the \icrcAbbr we have identified the requirements that a digital aid-distribution solution must fulfill to guarantee that it preserves the safety, rights, and dignity of aid recipients.
Then, we propose a decentralized aid-distribution system that enables humanitarian organizations to distribute physical goods at scale in a secure and privacy-preserving manner, while providing strong accountability.

We provide two instantiations of our design on two kinds of digital tokens already used in humanitarian contexts: smart cards and smartphones.
We formally prove that our schemes provide the required security and privacy properties, and we empirically demonstrate that, despite the use of advanced cryptography, they are efficient enough to support mass aid-distribution.

Our interactions with the \icrcAbbr reveal that, due to them dealing with the most vulnerable populations, the needs of humanitarian organizations often cannot be satisfied with off-the-shelf commercial solutions. 
Their requirements often bring up challenging research questions and open new design spaces rarely explored by our community.
We hope that our work fosters new collaborations between researchers and humanitarian organizations to explore this space so research innovations can directly benefit those most in need of help.

\section*{Acknowledgements}
\addcontentsline{toc}{section}{Acknowledgments}

We thank Laurent Girod for implementing the Java Card prototype and Nathan Duchesne for implementing the smartphone prototype.

This work was partially funded by the Science and Technology for Humanitarian Action Challenges (HAC) programme from the Engineering for Humanitarian Action initiative, a partnership between the ICRC, EPFL and ETHZ.
 
\bibliographystyle{IEEEtranS}
\bibliography{sources}

\appendices

\crefalias{section}{appendix}

\section{Privacy-friendly Smartphone-Station Channel}
\label{subsec:channel}
When using a smart card as a token, the communication between the token and the stations happens through a card reader. 
Such a communication channel is hard to eavesdrop, and as long as the card reader is honest, hard to compromise.
When using a smartphone as a token, the communication becomes wireless.
Wireless communications are is easy to eavesdrop on and, depending on the protocol, are also susceptible to person-in-the-middle attacks.
This new attack surface can have a strong impact on privacy.

We discuss four widely used communication channels suitable for connecting the smartphone to the station when deploying the system: cellular network, local Wi-Fi, Bluetooth, and QR-code scanning.
We compare them in terms of (1) ease of integration into a humanitarian aid-distribution system and (2) the privacy risks they introduce. 

\parBF{Cellular Network} 
Smartphones can use cellular networks to connect to the Internet. 
If the stations also have Internet connection, the phone can communicate with the station via this channel.
Using cellular networks has the advantage that it does not require the installation of dedicated hardware to support the communication.
However, one cannot always assume that the areas where humanitarian aid distribution takes place have reliable and high-bandwidth cellular connection system.

The use of cellular networks extends the threat model of the system to include the mobile service provider, as well as eavesdroppers with adequate equipment~\cite{RupprechtDHWP18}. These adversaries can link recipients, and hence, breach their privacy.

\parBF{Local Wi-Fi} 
An alternative for connecting smartphones and stations is to use a local Wi-Fi network to which both devices connect. This method requires the deployment and maintenance of Wi-Fi routers. Once these routers are deployed, quality of connection is easy to achieve.

Using a Wi-Fi eliminates the need to trust a mobile service provider. Instead, trust is put on the entity setting up the router to not track users~\cite{GaoL0QCQLGL21,YuLMZL20}. 
If this entity is not trustworthy, users need to have unlinkable MAC addresses, which are not available in all commercial devices~\cite{MartinMDFBRRB17}, and prevent fingerprinting of other information contained in the probe request~\cite{VanhoefMCCP16, GentryP16,FenskeBMMRR21}.

\parBF{Bluetooth} 
A third option is to use the Bluetooth Low Energy (BLE) technology~\cite{BaruaAHH22, CasarPST22}. It enables smartphones to connect to the station without intermediaries that could be adversarial.
Compared to Wi-Fi, BLE provides less bandwidth, but is still enough to run our protocols. On the negative side, BLE does not have the same prevalence on phones as Wi-Fi.

With regard to privacy, the BLE specification contains provisions to randomize the MAC address of the devices. 
However, enabling full randomization to ensure anonymity from the receiving end comes at the cost of having to pair the devices every time. 
Whether frequent re-pairing is acceptable depends on the periodicity of the distribution.

\parBF{QR-code Scanning}
QR codes are two-dimensional bar codes containing data~\cite{ISO18004}. Users can access the encoded data by scanning the QR code with the camera of their smartphones~\cite{KrombholzFKKHW14}. Due to their ease of deployment, QR codes are increasingly popular in smartphone settings~\cite{WahshehL20}, e.g., to provide links to social media accounts, to show vaccination certificates, or to share files. 

Because it requires the adversary to have direct line of sight to the QR code to capture the encoded data~\cite{WahshehL20}, QR-code scanning provides a high level of privacy. However, the storage capacity of a QR code is limited (a maximum of 3 kB~\cite{ISO18004}). Thus, more than one QR code may be needed to transmit all data necessary in our protocols. 
Even though we can rotate QR codes to transfer more data, the throughout is still low due to other constraints, e.g, the recognition delay of the phone, people unlocking the phone to point the camera to the code, etc. This may slow down the protocol's execution rendering this method impractical in reality (\reqlink{scale}).

\begin{table}[tbp]
    \centering
    \caption{Comparing Smartphone-Station Channel}
    \begin{tabular}{@{}l l l l l@{}} 
     \toprule
      & Cellular& Local & BLE & QR-code \\ 
      & Network & Wi-Fi &  &  \\ 
     \midrule
     Privacy & low & medium & medium & high\\   
     Throughput & high & high & medium & low \\
Infrastructure & \xmark & \cmark & \xmark & \cmark\\
     User-friendliness & \cmark & \cmark & \cmark & \xmark \\
     \bottomrule
    \end{tabular}
    \label{tab:channeldis}
\end{table}

Table~\ref{tab:channeldis} summarizes the performance of the four channels. If the blocklist is short and only requires low throughput of the channel, QR-code scanning can be a good option because of better privacy. If the blocklist is longer and requires larger throughput, BLE or local Wi-Fi would be a more desirable solution, depending on the ability of setting up Wi-Fi device. BLE has an advantage over Wi-Fi if devices can enable fully randomization of MAC addresses. Cellular network should be the last alternative due to the large privacy risk. In the case where the privacy risk is unacceptable, one option is to move to smart-card-based solution, another option is to shorten the distribution period to encourage having a shorter blocklist.
 \section{Cryptography Details}
\subsection{Digital Signatures}
Let $\Pi = (\dsgen,\allowbreak \dssign,\allowbreak \dsverify)$ be a signature scheme. We restate the existential unforgeability game $\code{Sig-forge}_{\adv, \Pi}(\ell)$ for adversary $\adv$ and security parameter $\ell$:
\begin{enumerate}
    \item The challenger runs $\dsgen(\secl)$ to obtain a signing-verification key-pair $(\sk, \pk)$.
    \item The adversary \adv is given $\pk$ and has access to the signing oracle $\dssign(\sk, \cdot)$.
    \item Eventually \adv outputs a forgery $(\sigma, m)$. Let $\mathcal{Q}$ denote the set of all message queries that \adv asked its oracle. The adversary succeeds if and only if (1) $\dsverify(\pk, \sigma, m) = \top$ and (2) $m \notin \mathcal{Q}$.
    \item The experiment outputs 1 if the adversary wins.
\end{enumerate}

A signature scheme $\Pi = (\dsgen,\allowbreak \dssign,\allowbreak \dsverify)$ is \emph{unforgeable} if for all probabilistic polynomial-time adversaries \adv, there is a negligible function \code{negl} such that: 
\[
   \prob \left[ \code{Sig-forge}_{\adv, \Pi}(n) = 1 \right] \leqslant \code{negl}(n).
\]

\subsection{Pseudorandom Functions}
Pseudorandom functions (PRFs) are ``random-looking functions'' which refer to the pseudorandomness of a distribution on functions~\cite{KatzLindell2014}. 
The set $\code{Func}_{n}$ are all functions mapping n-bit strings to n-bit strings.

\begin{definition}
    An efficient, length-preserving, keyed function $F: \{0, 1\}^{n} \times \{0, 1\}^{n} \rightarrow \{0, 1\}^{n}$ is a pseudorandom function if no probabilistic polynomial-time distinguisher $D$ can differentiate $F$ from $f$ such that $f \in \code{Func}_{n}$, i.e., 
    \[
    |\prob\,[D^{F_{k}(\cdot)}(1^{n}) = 1] - \prob\,[D^{f(\cdot)}(1^{n}) = 1]| \leqslant \code{negl}(n), 
    \] 
    where the first probability is taken over uniform choice of $k \in_{R} \{0, 1\}^{n}$ and the randomness of $D$. 
\end{definition}

\subsection{Pointcheval-Sanders Credentials}
\label{app:ps-abc}

There are three parties in an ABC scheme: \emph{the issuer}, \emph{the user}, and \emph{the verifier}. 
The issuer sets up the system and issues credentials to users. 
The process by which a user obtains a credential is called \emph{issuance}. 
The user holds credentials and shows them to the verifier. 
The user can choose to reveal some attributes from any number of the credentials to the verifier. 
The verifier checks the credential is valid, and the revealed attributes fulfill the requirements of the application. 
The process by which the user shows possession of a credential to the verifier is called \emph{verification}. 

A secure ABC scheme has the following properties: 
\begin{itemize}
    \item Unforgeability: It is not possible for any party in the system to forge a credential without the help of the issuer. 
    \item Unlinkability: It is not possible for the verifier (even when colluding with the issuer) to distinguish between two users who disclose the same attributes in the verification process. 
\end{itemize}

\subsubsection{Setup}
The issuer runs \code{ABC.Gen} to set up global parameters and generate keys. It proceeds as follows:

\begin{enumerate}
  \item The algorithm takes as input the security parameter $\secl$ and the number of attributes $L$.
  \item It generates public parameters containing a type-III bilinear group pair given by $\paramps = (e(\cdot,\allowbreak \cdot),\allowbreak \group_1,\allowbreak \group_2,\allowbreak \group_T,\allowbreak \grouporder,\allowbreak g_1,\allowbreak g_2,\allowbreak g_T)$ where $g_1, g_2, g_T$ are generators of the groups $\group_1, \group_2, \group_T$ of order $\grouporder$ respectively.
  \item It generate signing and verification keys. It picks $x, y_1, \dots, y_L \in_{R} \mathbb{Z}_{q}$, $g \in_{R} \mathbb{G}_{1}$ and $\tilde{g} \in_{R} \mathbb{G}_{2}$, then computes: 
    \[
    X = g^x, \tilde{X} = \tilde{g}^x, Y_i = g^{y_i}, \tilde{Y}_{i} = \tilde{g}^{y_i}
    \]
    for $i = 1, \dots, L$. It publishes the public key \pkps and returns the private key \skps, where 
    \begin{align*}
    \pkps &= (g, Y_1, \dots, Y_L, \tilde{g}, \tilde{X}, \tilde{Y}_{1}, \dots, \tilde{Y}_{L}) \\
    \skps &= (x, X, y_1, \dots, y_L).
    \end{align*}
\end{enumerate}

\subsubsection{Issuance}\label{app:abc:issue}
The user and the issuer jointly run \code{ABC.issue} protocol to create signatures as credentials. 
The issuer acts as a signer without knowing all the attributes it is signing. The protocol proceeds as follows:
\begin{enumerate}
    \item \textit{Input agreement.} The user and the issuer agree on the set of attribute indices $\mathcal{I} \subset \{1, \dots, L\}$ that are determined by the issuer and the set of attribute indices $\mathcal{U} \subset \{1, \dots, L\}$ that are determined by the user, where $\mathcal{I} \cup \mathcal{U} = \{1, \dots, L\}$. The user takes as input the public key \pkps and the attributes $a_i, \forall i \in \mathcal{U}$. The issuer takes as input the public key \pkps, the private key \skps, and the attributes $a_i, \forall i \in \mathcal{I}$. 
    
    \item \textit{User commitment.} The user commits to the attributes they want to include in the credential by picking $t \in_{R} \mathbb{Z}_{q}$ at random and computing the credential $C$ as well as a non-interactive zero-knowledge proof (NIZK) $\pi$ that proves $C$ has been computed correctly:
    \begin{align*}
    C &= g^{t} \prod^{}_{i \in \mathcal{U}} Y_{i}^{a_{i}}, \\
    \pi &= \nizk\left\{(t, (a_i)_{i \in \mathcal{U}}): C = g^{t} \prod^{}_{i \in \mathcal{U}} Y_{i}^{a_{i}} \right\}.
    \end{align*}
    The user sends an issuance request $\req_{\text{abc}} = (C, \pi)$ to the issuer.  
    
    \item \textit{Issuer signing.} The issuer receives the issuance request $\req_{\text{abc}}\allowbreak = (C,\allowbreak \pi)$, verifies the validity of proof $\pi$ with respect to commitment $C$, and aborts if the proof is not correct. Otherwise, the issuer picks $u \in_{R} \mathbb{Z}_{q}$ at random and creates the signature 
    \[
    \sigma' = \left( g^{u}, \left( XC \prod_{i\in \mathcal{I}} Y_{i}^{a_i}\right)^{u}\right).
    \]
    The issuer sends $\sigma'$ and the attributes chosen by the issuer $a_i, \forall i \in \mathcal{I}$ to the user. 
    
    \item \textit{Unblinding signature.} The user receives the signature $\sigma' = (\sigma_{1}', \sigma_{2}')$ and attributes. 
The user computes the ``unblinding'' signature
    \[
    \sigma = \left( \sigma_{1}', \frac{\sigma_{2}'}{(\sigma_{1}')^{t}}\right), 
    \]
    and uses the public key \pk to validate that $\sigma$ is a valid signature on the attributes. If valid, the user stores both $\sigma$ and the attributes. 
\end{enumerate}

When using the ABC scheme in our protocols for the smartphone solution, the phone additionally proves that the revocation value $\RevokeId = (\generator_{1}, \generator_{1}^{\RevokeVal}) = (\generator_{1}, \overline{\RevokeId})$ has been correctly formed. The issuer can trivially check that the first component equals the generator of $\group_1$. The token extends the proof of correctness $\pi$ above to show that the second component is correct. Instantiating the user-defined attributes at positions 1 and 3, it instead proves:
\begin{align*}
\pi &= \nizk\left\{(t, \kH, \RevokeVal): C = \generator^{t} Y_{1}^{\kH} Y_{3}^{\RevokeVal} \land \overline{\RevokeId} = \generator_{1}^{\RevokeVal} \right\}
\end{align*}

\subsubsection{Verification}
The user can run \code{ABC.show} to convince the verifier that the user possesses a valid credential over a set of attributes using a zero-knowledge proof. 
As part of this proof, the user can choose to reveal some of the attributes while hiding others from the verifier. 
The verifier checks and accepts the proof if all the checks succeed. 
The showing protocol proceeds as follows:
\begin{enumerate}
    \item \textit{Input agreement.} The user and the verifier agree on the public key of the issuer and the set of attributes $\mathcal{D}$ that should be disclosed to the verifier. 
    Let $\mathcal{H} = \{1, \dots, L\} \setminus \mathcal{D}$ be the set of attributes that are hidden from the verifier. 
    
    \item \textit{Proof creation.} The user takes as input a signature $\sigma = (\sigma_1, \sigma_2)$ over the attributes $a_1, \dots, a_L$ and picks random values $r_s, t_s \in_{R} \mathbb{Z}_{q}$. 
    The user computes a randomized signature $\sigma_s$ and a non-interactive zero-knowledge proof that proves the signature $\sigma_s$ is a valid signature:   
    \[
    \sigma_s = (\sigma_{s_{1}}, \sigma_{s_{2}}) = (\sigma^{r_s}_1, (\sigma_2 \sigma_{1}^{t_s})^{r_s}),
    \]
    \begin{equation*}
    \begin{aligned}
        \pi_s = \nizk\{ (t_s, (a_i)_{i \in \mathcal{H}}) & : \frac{e(\sigma_{s_{2}}, \tilde{g})\prod_{i \in \mathcal{D}} e(\sigma_{s_{1}}, \tilde{Y}_{i})^{-a_{i}}}{e(\sigma_{s_{1}}, \tilde{X})} \\
        & = e(\sigma_{s_{1}}, \tilde{g})^{t_s} \prod_{i \in \mathcal{H}} e(\sigma_{s_{1}}, \tilde{Y}_{i})^{a_i}. \}
    \end{aligned}
    \end{equation*}
    The user sends $(\sigma_s, (a_i)_{i \in \mathcal{D}}, \pi_s)$ to the verifier.

    \item \textit{Proof verification.} The verifier receives the signature $\sigma_s = (\sigma_{s_{1}}, \sigma_{s_{2}})$, the disclosed attributes $(a_i)_{i \in \mathcal{D}}$, and the proof $\pi_s$ from the user. 
    First, the verifier checks that $\sigma_{s_{1}} \neq 1_{\mathbb{G}_{1}}$, i.e., $\sigma_{s_{1}}$ is not the unity element in $\mathbb{G}_{1}$. 
    Second, the verifier checks the user has a valid signature under the public key $\pkps$ over the disclosed attributes $(a_i)_{i \in \mathbb{D}}$ by verifying the proof $\pi_s$. 
    The verifier accepts the disclosure proof only if both checks pass. 
\end{enumerate}

\begin{fullversion}
\end{fullversion}

\section{Card-whispering Protocol}\label{app-cardwhisper}

Recall that to provide robust distribution (\reqlink{robust}) it is necessary to enable more than one member of a household to retrieve their entitlement.
This is needed to support household members losing their cards or unable to attend the distribution due to sickness. 
Our solution achieves robustness by permitting that members of a household have ``cloned'' cards, i.e., cards that share the same state.

To create a card clone, registered household members need to bring their card to the registration station. 
Then, they run the following card-whispering protocol between a \emph{parent card} -- the card that the household member wants to clone, which contains the household secret and current state, and a \emph{child card} -- the card that will become a clone, which is initialized in the setup with the built-in shared keys and public parameters but does not contain any secret.
The parent card and child card are inserted into dedicated devices. The protocol goes as follows:

\begin{itemize}
    \item $\CardWhispering{}$ 
      The card whispering protocol is between a parent and a child card. The parent card takes as input its state $\stateT^{P} = (\lastepoch,\allowbreak \pk,\allowbreak \sk,\allowbreak \kH,\allowbreak \ent,\allowbreak \RevokeId)$, the client card has no input. The parent card sends $\stateT^{P}$ to the client. The client sets its own state $\stateT^{C}$ to equal $\stateT^{P}$. The parent card does not modify its state.
\end{itemize}

After running the protocol, the child card will have the same state as the parent card (including the same secret of the household). We have described the protocol between two cards, but it can be run among any number of cards.

As explained in \cref{sec:sc-based}, to maintain privacy and security this protocol must be conducted in a controlled environment where the owner of the parent card can be authenticated and requires (1) that only the child card learns the household secret (otherwise privacy is at risk, see \cref{sec:sc-based}), and (2) that the child card is given to a member of the household (otherwise security is at risk, as illegitimate recipients may get the entitlement of the household).

\end{document}